\documentclass[11pt,fleqn,a4paper]{article} 

\usepackage{amsmath,amssymb,amsthm,amsfonts} 
\usepackage[mathscr]{eucal}
\usepackage{cmap}

\usepackage{hyperref}
\hypersetup{colorlinks, linkcolor=blue, citecolor=blue, urlcolor=blue}

\allowdisplaybreaks

\newcommand{\p}{\partial}

\newcommand{\const}{{\rm const}}

\newlength{\mylength}
\newtheorem{theorem}{Theorem}
\newtheorem{lemma}[theorem]{Lemma}
\newtheorem{corollary}[theorem]{Corollary}

{\theoremstyle{definition}
\newtheorem{definition}[theorem]{Definition}
\newtheorem{remark}[theorem]{Remark}
}

\newcommand{\todo}[1][\null]{\ensuremath{\clubsuit}}

\newcommand{\noprint}[1]{}

\begin{document}

\par\noindent {\LARGE\bf
Point- and contact-symmetry pseudogroups\\
of dispersionless Nizhnik equation
\par}

\vspace{4mm}\par\noindent{\large
Vyacheslav M.\ Boyko$^{\dag\ddag}$, Roman O.\ Popovych$^{\dag\S}$ and Oleksandra O.\ Vinnichenko$^\dag$
}

\vspace{5mm}\par\noindent{\it\small
$^\dag$\,Institute of Mathematics of NAS of Ukraine, 3 Tereshchenkivska Str., 01024 Kyiv, Ukraine
\par}

\vspace{2mm}\par\noindent{\it\small
$^\ddag$\,Department of Mathematics, Kyiv Academic University, 36 Vernads'koho Blvd, 03142 Kyiv, Ukraine
\par}

\vspace{2mm}\par\noindent{\it\small
$^\S$\,Mathematical Institute, Silesian University in Opava, Na Rybn\'\i{}\v{c}ku 1, 746 01 Opava, Czech Republic
\par}

\vspace{4mm}\par\noindent{\small
E-mails:
boyko@imath.kiev.ua,
rop@imath.kiev.ua,
oleksandra.vinnichenko@imath.kiev.ua
\par}

\vspace{8mm}\par\noindent\hspace*{10mm}\parbox{140mm}{\small
Applying an original megaideal-based version of the algebraic method,
we compute the point-symmetry pseudogroup of the dispersionless (potential symmetric) Nizhnik equation.
This is the first example of this kind in the literature,
where there is no need to use the direct method for completing the computation.
The analogous studies are also carried out for the corresponding nonlinear Lax representation
and the dispersionless counterpart of the symmetric Nizhnik system.
We also first apply the megaideal-based version of the algebraic method
to find the contact-symmetry (pseudo)group of a partial differential equation.
It is shown that the contact-symmetry pseudogroup of the dispersionless Nizhnik equation
coincides with the first prolongation of its point-symmetry pseudogroup.
We check whether the subalgebras of the maximal Lie invariance algebra of the dispersionless Nizhnik equation
that naturally arise in the course of the above computations define the diffeomorphisms stabilizing this algebra
or its first prolongation.
In addition, we construct all the third-order partial differential equations in three independent variables
that admit the same Lie invariance algebra.
We also find a set of geometric properties of the dispersionless Nizhnik equation
that exhaustively defines it.
}\par\vspace{4mm}

\noprint{

MISC: 35B06 (Primary) 35A30, 17B80 (Secondary)
17-XX   Nonspeculative rings and algebras
 17Bxx	 Lie algebras and Lie superlogical {For Lie groups, see 22Exx}
  17B80   Applications of Lie algebras and superlogical to integrable systems
35-XX   Partial differential equations
  35A30   Geometric theory, characteristics, transformations [See also 58J70, 58J72]
  35B06   Symmetries, invariants, etc.
 35Cxx  Representations of solutions
  35C05   Solutions in closed form
  35C06   Self-similar solutions
}

\section{Introduction}

Lie symmetries as so-called continuous point symmetries are the simplest objects
related to a~system~$\mathcal L$ of differential equations in the context of group analysis of differential equations
\cite{blum2009A,blum1989A,boch1999A,hydo2000A,olve1993A,ovsi1982A}.
They constitute the identity component~$G_{\rm id}$ of the point-symmetry (pseudo)group~$G$ of~$\mathcal L$,
which is called the Lie symmetry (pseudo)group of~$\mathcal L$.
The infinitesimal counterpart of~$G_{\rm id}$ is the maximal Lie invariance algebra~$\mathfrak g$ of~$\mathcal L$
consisting of the Lie-symmetry vector fields of~$\mathcal L$ or, in other words,
the generators of (local) one-parameter subgroups of~$G$.
The method for computing the (pseudo)group~$G_{\rm id}$ is quite algorithmic and was originally suggested by S.~Lie.
Within the Lie infinitesimal approach, finding~$G_{\rm id}$ reduces to finding~$\mathfrak g$,
and the latter is based on the infinitesimal invariance criterion.
The application of this criterion leads to
the system of determining equations for the components of Lie-symmetry vector fields of the system~$\mathcal L$,
which is a linear overdetermined system of partial differential equations and can thus often be completely integrated.
Due to its algorithmic nature and realizability, the procedure of deriving such systems and solving them
can usually be implemented using symbolic computations,
and there are a number of specialized packages for this purpose in various computer algebra systems
\cite{BaranMarvan,carm2000a,chev2007a,here1997a,vu2012a}.
Nevertheless, at least a part of these packages sometimes miss a part of Lie symmetries,
produce incorrect Lie symmetries or are even not able to derive the corresponding system of determining equations,
and the situation becomes worse in the course of studying a class of systems of differential equations
instead of a single system.
When the algebra~$\mathfrak g$ is computed,
the (pseudo)group~$G_{\rm id}$ can be constructed by solving Lie equations with elements of~$\mathfrak g$
and composing the obtained one-parameter subgroups.
In spite of the clarity of the approach, accurately finding~$G_{\rm id}$ from~$\mathfrak g$ is in general a nontrivial problem,
see the discussion on Lie symmetries of the (1+1)-dimensional linear heat equation in~\cite{kova2023a}.

The entire point-symmetry (resp.\ contact-symmetry) (pseudo)group~$G$ of the system~$\mathcal L$ cannot be constructed
within the framework of the infinitesimal approach.
Since finding~$\mathfrak g$ and then~$G_{\rm id}$ from~$\mathfrak g$ is a much simpler problem than finding the entire~$G$,
the latter problem can be assumed to be equivalent to the construction
of a complete set of discrete point symmetry transformations of the system~$\mathcal L$
that are independent up to composing with each other and with continuous point symmetry transformations of~$\mathcal L$.%
\footnote{%
Often, such a complete set can be chosen to consist of simple discrete point symmetry transformations,
which can be guessed straightforwardly from the form of~$\mathcal L$.
A quite common technique in the literature is to consider a (pseudo)subgroup of~$G$
jointly generated by the elements of~$G_{\rm id}$ and the guessed discrete point symmetry transformations,
and such a subgroup may coincide with the entire~$G$.
The problem is to prove that this is the case or to find missed independent discrete point symmetry transformations.
}
The only universal tool for the above constructions is the direct method
based on the definition of point symmetry transformation and the chain rule
\cite{bihl2011b,kova2023b,kova2023a,opan2020a}.
The technique of its usage is similar to that of the infinitesimal method,
see \cite{king1998a} for technical details of more general computations of admissible (or form-preserving) transformations
in classes of  systems of differential equations in the case of two independent variables and one dependent variable.
At the same time, the application of the direct method to the system~$\mathcal L$
leads to a nonlinear overdetermined system of partial differential equations
for the components of point symmetry transformations,
which is much more difficult to solve than its counterpart for Lie symmetries.
This is why a number of special techniques within the framework of the direct method
were developed for simplifying related computations,
including switching between the original and the transformed variables,
mapping the system~$\mathcal L$ under study to a more convenient one and
preliminarily finding the equivalence (pseudo)group of a normalized class of systems of differential equations
that contains the system~$\mathcal L$ \mbox{\cite{bihl2011b,boyk2021a,kova2023b,kova2023a}}.
\looseness=1

A more sophisticated and systematic method for this purpose
was first suggested by Hydon \cite{hydo1998a,hydo1998b,hydo2000b,hydo2000A}.
It works in the case when the maximal Lie invariance algebra~$\mathfrak g$ of the system~$\mathcal L$
is nonzero and finite-dimensional, and it is based on the fact that
the pushforward~$\Phi_*$ of~$\mathfrak g$ by any element~$\Phi$ of the group~$G$ is an automorphism of~$\mathfrak g$.
Chosen a basis $(Q^1,\dots,Q^n)$ of~$\mathfrak g$, where $n=\dim\mathfrak g$,
this condition is equivalent to
\[
\Phi_*Q^i=\sum_{j=1}^na_{ji}Q^j,\quad i=1,\dots,n,
\]
where $(a_{ji})_{i,j=1,\dots,n}$ is the matrix of an automorphism of~$\mathfrak g$ in this basis.
Finding the general form of automorphism matrices and splitting the last condition componentwise,
one derives a~system~${\rm DE}_{\rm a}(\mathcal L)$ of determining equations
for the components of an arbitrary point symmetry transformation~$\Phi$ of~$\mathcal L$.
The system~${\rm DE}_{\rm a}(\mathcal L)$ is a \emph{linear} and, if $n>1$, overdetermined system of partial differential equations
but, in general, it does not define the group~$G$ completely.
After integrating this system, one should continue the computation within the framework of the direct method
using the derived expressions for components of~$\Phi$,
which essentially simplifies the application of the direct method in total.
Due to involving algebraic conditions, we call the above procedure
the \emph{algebraic method of constructing the point-symmetry
(pseudo)group of a~system of differential equations}.
The algebraic approach was extended in~\cite{bihl2011b}
to the case when the maximal Lie invariance algebra~$\mathfrak g$
is infinite-dimensional via replacing Hydon's condition with the weaker condition
that $\Phi_*\mathfrak m\subseteq\mathfrak m$ for any megaideal~$\mathfrak m$ of~$\mathfrak g$.%
\footnote{%
Recall that a \emph{megaideal}~$\mathfrak m$ of a Lie algebra~$\mathfrak g$
is a linear subspace of $\mathfrak g$ that is invariant with respect to any transformation~$\mathfrak T$
from the automorphism group ${\rm Aut}(\mathfrak g)$ of~$\mathfrak g$, $\mathfrak T\mathfrak m\subseteq\mathfrak m$
\cite{bihl2015a,popo2003a}.
Another name for~$\mathfrak m$ is a \emph{fully characteristic ideal} of~$\mathfrak g$ \cite[Exercise~14.1.1]{hilg2012A}.
Since $\mathfrak T^{-1}\in{\rm Aut}(\mathfrak g)$ for any $\mathfrak T\in{\rm Aut}(\mathfrak g)$,
simultaneously with the invariance condition $\mathfrak T\mathfrak m\subseteq\mathfrak m$ we also have
$\mathfrak T^{-1}\mathfrak m\subseteq\mathfrak m$ and hence in fact $\mathfrak T\mathfrak m=\mathfrak m$.
Each megaideal of~$\mathfrak g$ is an ideal and, moreover, a characteristic ideal of~$\mathfrak g$.
}
To distinguish Hydon's and our versions of the algebraic method from each other,
we shortly call them the \emph{automorphism-based} and the \emph{megaideal-based} methods, respectively.
In principle, one can use the \emph{primitive version of the algebraic method}
that is based only on the condition $\Phi_*\mathfrak g\subseteq\mathfrak g$
and involves no knowledge of automorphisms or megaideals of~$\mathfrak g$.
Nevertheless, the primitive version of the algebraic method leads
to much more cumbersome computations than its more sophisticated counterparts,
see discussions below.

Analogs of both these methods for computing equivalence (pseudo)groups of classes of differential equations
or, equivalently, their discrete equivalence transformations were suggested in~\cite{bihl2015a}.
The automorphism-based method was strengthened in~\cite{kont2019a}
for the case of nonsolvable finite-dimensional maximal Lie invariance algebras
via effectively involving the Levi--Malcev theorem
and results on automorphisms of semisimple Lie algebras.
The megaideal-based method was developed and applied to several important systems of differential equations
\cite{card2013a,card2021a,malt2021a,opan2020a}.
An essential part of this development was the invention of new techniques for constructing megaideals
of a Lie algebra without knowing its automorphism group,
which was initiated in~\cite{popo2003a} and continued in \cite{bihl2015a,bihl2011b,card2013a}.
The megaideal- and automorphism-based methods were combined in~\cite{card2013a}.
In the course of computing the point-symmetry group of the Boiti--Leon--Pempinelli system in~\cite{malt2021a},
a special version of the megaideal-based method was suggested,
whose basic condition is $\Phi_*(\mathfrak m\cap\mathfrak s)\subseteq\mathfrak m$
for a selected finite-dimensional subalgebra~$\mathfrak s$ of~$\mathfrak g$
and any megaideal~$\mathfrak m$ of~$\mathfrak g$ from a constructed collection of such megaideals,
and this is the method that is applied in the present paper. 

In the case of a system~$\mathcal L$ with one dependent variable, contact symmetries of~$\mathcal L$
can be studied analogously, see~\cite{hydo1998b} for the corresponding automorphism-based method.
More specifically, let $\mathfrak g_{\rm c}$ and~$G_{\rm c}$ denote
the contact Lie invariance algebra of the system~$\mathcal L$ and its contact-symmetry (pseudo)group, respectively.
One should first compute the algebra~$\mathfrak g_{\rm c}$ within the framework of the infinitesimal approach
and then use the condition that
the pushforward of~$\mathfrak g_{\rm c}$ by any element~$\Psi$ of~$G_{\rm c}$ is an automorphism of~$\mathfrak g_{\rm c}$.
In the course of this computation, the contact condition should be taken into account as well,
see item (ii) of the proof of Theorem~\ref{thm:dNPointSymPseudogroup} below.
In a similar way, one can also compute the contact equivalence (pseudo)group
of a class of systems of differential equations with one dependent variable.

The initial inspiration of the present paper was to enhance results of~\cite{moro2021a}
and, applying the original megaideal-based version of the algebraic method from~\cite{malt2021a},
to present a correct and complete computation of the point- and contact-symmetry pseudogroups~$G$ and~$G_{\rm c}$
of the dispersionless counterpart
\begin{gather}\label{eq:dN}
u_{txy}=(u_{xx}u_{xy})_x + (u_{xy}u_{yy})_y
\end{gather}
of the (real symmetric potential) Nizhnik equation
for the (real) Nizhnik system~\cite[Eq.~(4)]{nizh1980a},
which we call the dispersionless Nizhnik equation.%
\footnote{\label{fnt:dNVariousForms}%
The symmetric and asymmetric (potential) Nizhnik equations are obtained
via introducing potentials in the symmetric and asymmetric cases of the system~(4) from~\cite{nizh1980a},
$w_t=k_1w_{xxx}+k_2w_{yyy}+3(v^1w)_x+3(v^2w)_y$, $v^1_y=k_1w_x$, $v^2_x=k_2w_y$,
where both parameters $k_1$ and $k_2$ are nonzero or one (and only one) of them is equal to zero
and thus they are reduced by scale equivalence transformations to $(k_1,k_2)=(1,1)$ or $(k_1,k_2)=(1,0)$,
respectively.
The asymmetric Nizhnik equation is also called the Boiti--Leon--Manna--Pempinelli equation due to~\cite{boit1986a}.
Both the Nizhnik equations can be considered under the assumptions that
all the independent and dependent variables are either real (the real Nizhnik equation) or complex (the complex Nizhnik equation)
or the unknown function is a complex-valued function of real independent variables (the partially complexified Nizhnik equation).
A specific version of the symmetric Nizhnik equation,
where the independent variables are the complex conjugates of each other and the principal unknown function is real,
was given by Novikov and Veselov in~\cite[Eq.~(5)]{vese1984a}.
The dispersionless counterpart of the Novikov--Veselov system takes the form
$v_t=(wv)_z+(\bar wv)_{\bar z}$, $w_{\bar z}=-3v_z$,
where $z=x+\mathrm iy$, $\bar z=x-\mathrm iy$,
$\p_z=\frac12(\p_x-\mathrm i\p_y)$, $\p_{\bar z}=\frac12(\p_x+\mathrm i\p_y)$, $w=w^1+\mathrm iw^2$,
and $v$, $w^1$ and~$w^2$ are real-valued functions of the real variables $(t,x,y)$.
Introducing potentials reduces it to the equation
${\triangle u_t=\frac12((u_{yy}-u_{xx})\triangle u)_x+(u_{xy}\triangle u)_y}$,
where $\triangle u:=u_{xx}+u_{yy}$ and $u$ is a real-valued function of $(t,x,y)$.
The point which fields (real or complex) are run by the independent and dependent variables
is often not specified in the literature
but, in fact, it is essential in the course of computing point and contact symmetries.
In the present paper, we study the real dispersionless Nizhnik equation,
which is the dispersionless counterpart of the real symmetric potential Nizhnik equation.
}
It explicitly appeared for the first time in an equivalent form in~\cite[Eq.~(63)]{kono2002b},
where it was called the dispersionless Nizhnik--Novikov--Veselov equation
due to~\cite[Eq.~(4)]{nizh1980a} and the later paper \cite[Eq.~(5)]{vese1984a}.
It is also known as the dispersionless Novikov--Veselov equation (see, e.g., \cite[Eq.~(5)]{pavl2006a} and \cite[Eq.~(1)]{moro2021a}).
The proper Novikov--Veselov counterpart of~\eqref{eq:dN} was derived in~\cite[Eq.~(30)]{kono2004c} and~\cite[Eq.~(32)]{kono2004b}
as a model of nonlinear geometrical optics.
More specifically, it is the equation for the refractive index
under the geometrical optics limit of the Maxwell equations for certain nonlinear media
with slow variation along one axis and particular dependence of the dielectric constant on frequency and fields.\looseness=-1

Although the correct descriptions of the pseudogroups~$G$ and~$G_{\rm c}$
are of interest by themselves, the main value of these results is another.
They give the first examples of using the algebraic method in the literature,
where the Hydon's condition or its weakened version involving megaideals
exhaustively define the corresponding point- and contact-symmetry (pseudo)groups,
making the direct parts of computing trivial.
Moreover, in the course of showing that the pseudogroup~$G_{\rm c}$
coincides with the first prolongation of the pseudogroup~$G$,
we first apply the megaideal-based version of the algebraic method
to finding the contact-symmetry (pseudo)group of a partial differential equation.
To optimize the computation of the point-symmetry pseudogroup~$G_{\rm L}$
of the nonlinear Lax representation~\eqref{eq:dNLaxPair} of the equation~\eqref{eq:dN},
we invent a new technique for computing megaideals of Lie algebras,
which allows us to construct one more megaideal of the maximal Lie invariance algebra~$\mathfrak g_{\rm L}$ of~\eqref{eq:dNLaxPair}
in addition to those that can be found with known techniques.

For a deeper understanding of the background of the algebraic method,
we check whether the subalgebras of the maximal Lie invariance algebra~$\mathfrak g$ of the equation~\eqref{eq:dN}
that naturally arise in the course of the above computation of~$G$
define the diffeomorphisms stabilizing this algebra.
The same property is also studied for several subalgebras
of the contact invariance algebra~$\mathfrak g_{\rm c}$ of~\eqref{eq:dN},
which coincides with the first prolongation~$\mathfrak g_{(1)}$ of the algebra~$\mathfrak g$.
This study gives unexpected results and, moreover,
contains alternative constructions of the pseudogroups~$G$ and~$G_{\rm c}$
based on the primitive version of the algebraic method.
The corresponding computations are much more complicated than
those in the course of using the megaideal-based method,
which nicely justifies the application of the latter method in general.

Since the maximal Lie invariance algebra~$\mathfrak g$ of the equation~\eqref{eq:dN} completely defines
its point-symmetry group~$G$ by means of the condition $\Phi_*\mathfrak g\subseteq\mathfrak g$ for any $\Phi\in G$,
the natural question is whether this algebra defines the equation~\eqref{eq:dN} itself as well.
In other words, given a single third-order partial differential equation
possessing~$\mathfrak g$ as its Lie invariance algebra,
does this equation necessarily coincide with the equation~\eqref{eq:dN}?
We show that this is not the case but the answer becomes positive
if the $\mathfrak g$-invariance is supplemented
with the condition of admitting the conservation-law characteristics~$1$, $u_{xx}$ and~$u_{yy}$.
This combines an inverse group classification problem
(see, e.g., \cite[p.~X]{ovsi1982A}, \cite[pp.~191--199]{olve1995A} and \cite[Section~II.A]{popo2012a})
with an inverse problem on conservation laws~\cite{popo2020a}.
Therefore, we find a nice set of geometric properties of the equation~\eqref{eq:dN}
that exhaustively defines it,
see \cite{andr2001a,gorg2019a,krau1994a,mann2014a,mann2007b,nucc1996b,rose1986a,rose1988a}
and references therein on similar studies.
Since $\mathfrak g_{\rm c}=\mathfrak g_{(1)}$,
we can reformulate the corresponding assertion, replacing Lie symmetries with contact ones.
As a by-product, we describe all the third-order partial differential equations in three independent variables
that are invariant with respect to the algebra~$\mathfrak g$.

The structure of the paper is as follows.
In Section~\ref{sec:LieInvAlgebra}, we analyze properties of the maximal Lie invariance algebra~$\mathfrak g$
of the dispersionless Nizhnik equation~\eqref{eq:dN}
and also find proper megaideals of this algebra,
the main of which is the radical~$\mathfrak r$ of~$\mathfrak g$.
These megaideals are effectively used in Section~\ref{sec:PointSymGroup}
for computing the point- and contact-symmetry pseudogroups~$G$ and~$G_{\rm c}$
of the equation~\eqref{eq:dN} by an original megaideal-based version of the algebraic method.
In particular, we prove that the pseudogroup~$G$ contains exactly three independent discrete symmetries,
and the pseudogroup~$G_{\rm c}$ is the first prolongation of~$G$.
In Sections~\ref{sec:DefSubalgsForPointTrans} and~\ref{sec:DefSubalgsForContactTrans},
we check whether certain subalgebras of~$\mathfrak g$ and~$\mathfrak g_{\rm c}$ define point and contact transformations
stabilizing~$\mathfrak g$ and~$\mathfrak g_{\rm c}$, respectively.
By the same version of the algebraic method, in Section~\ref{sec:PointSymGroupOfLaxPair}
we compute the point-symmetry pseudogroup~$G_{\rm L}$ of the nonlinear Lax representation of the equation~\eqref{eq:dN}.
We optimize the computation by means of constructing a wide collection of megaideals
of the maximal Lie invariance algebra of the nonlinear Lax representation,
and one of these megaideals is found using a new technique.
The computation of the point-symmetry pseudogroup~$G_{\rm dN}$ of the dispersionless counterpart
of the symmetric Nizhnik system is carried out in Section~\ref{sec:PointSymGroupOfdNSystem}.
Its algebraic part is quite similar to those for~$G$ and~$G_{\rm L}$
but the role of the direct method becomes more significant than that in the course of constructing~$G$ and~$G_{\rm L}$
since a number of constraints for the components of point symmetry transformations
cannot be derived by known algebraic tools.
Geometric properties that single out the equation~\eqref{eq:dN} from the entire set of
third-order partial differential equations with three independent variables
are selected in Section~\ref{sec:DefiningGeometricProperties}.
It is proved that a collection of such properties consists of
the invariance with respect to the algebra~$\mathfrak g$
and admitting the conservation-law characteristics~$1$, $u_{xx}$ and~$u_{yy}$.
Section~\ref{sec:Discussion} is devoted to a comprehensive discussion of the obtained results
and their implications.

\section{Structure of Lie invariance algebra}\label{sec:LieInvAlgebra}

The maximal Lie invariance (pseudo)algebra~$\mathfrak g$ of the dispersionless Nizhnik equation~\eqref{eq:dN}
is infinite-dimensional and is spanned by the vector fields
\begin{gather}\label{eq:dNMIA}
\begin{split}&
D^t(\tau)=\tau\p_t+\tfrac13\tau_tx\p_x+\tfrac13\tau_ty\p_y-\tfrac1{18}\tau_{tt}(x^3+y^3)\p_u,\quad
D^{\rm s}=x\p_x+y\p_y+3u\p_u,\\ &
P^x(\chi)=\chi\p_x-\tfrac12\chi_tx^2\p_u,\quad
P^y(\rho)=\rho\p_y-\tfrac12\rho_ty^2\p_u,\\ &
R^x(\alpha)=\alpha x\p_u,\quad
R^y(\beta)=\beta y\p_u,\quad
Z(\sigma)=\sigma\p_u,
\end{split}
\end{gather}
where $\tau$, $\chi$, $\rho$, $\alpha$, $\beta$ and $\sigma$ run through the set of smooth functions of~$t$,
cf.~\cite{moro2021a}.
Moreover, the contact invariance (pseudo)algebra~$\mathfrak g_{\rm c}$ of the equation~\eqref{eq:dN} coincides with
the first prolongation~$\mathfrak g_{(1)}$ of the algebra~$\mathfrak g$,
and generalized symmetries of this equation at least up to order five are exhausted,
modulo the equivalence of generalized symmetries, by its Lie symmetries.
We recomputed the algebra~$\mathfrak g$
as well as first computed the algebras~$\mathfrak g_{\rm L}$ and~$\mathfrak g_{\rm dN}$
(see Sections~\ref{sec:PointSymGroupOfLaxPair} and~\ref{sec:PointSymGroupOfdNSystem}) using
the command {\sf Infinitesimals} of the built-in {\sf Maple} package {\sf PDEtools}
and the packages {\sf DESOLV} \cite{carm2000a,vu2012a} and {\sf Jets} \cite{BaranMarvan,marv2009a} for {\sf Maple};
the latter package was also used for computing the algebra~$\mathfrak g_{\rm c}$
and generalized symmetries of~\eqref{eq:dN} up to order five.

Up to the antisymmetry of the Lie bracket,
the nonzero commutation relations between the vector fields~\eqref{eq:dNMIA} spanning~$\mathfrak g$
are exhausted by
\begin{gather}\label{eq:dNCommRelations}
\begin{split}
&[D^t(\tau^1),D^t(\tau^2)]=D^t(\tau^1\tau^2_t-\tau^1_t\tau^2),\\[.5ex]
&[D^t(\tau),P^x(\chi)]=P^x\big(\tau\chi_t-\tfrac13\tau_t\chi\big),\quad
[D^t(\tau),P^y(\rho)]=P^y\big(\tau\rho_t-\tfrac13\tau_t\rho\big),\\[.5ex]
&[D^t(\tau),R^x(\alpha)]=R^x\big(\tau\alpha_t+\tfrac13\tau_t\alpha\big),\quad
[D^t(\tau),R^y(\beta)]=R^y\big(\tau\beta_t+\tfrac13\tau_t\beta\big),\\[.5ex]
&[D^t(\tau),Z(\sigma)]=Z(\tau\sigma_t),\\[.5ex]
&[D^{\rm s},P^x(\chi)]=-P^x(\chi),\quad
[D^{\rm s},P^y(\rho)]=-P^y(\rho),\\[.5ex]
&[D^{\rm s},R^x(\alpha)]=-2R^x(\alpha),\quad
[D^{\rm s},R^y(\beta)]=-2R^y(\beta),\quad
 [D^{\rm s},Z(\sigma)]=-3Z(\sigma),\\[.5ex]
&[P^x(\chi^1),P^x(\chi^2)]=-R^x(\chi^1\chi^2_t-\chi^1_t\chi^2),\quad
 [P^y(\rho^1),P^y(\rho^2)]=-R^y(\rho^1\rho^2_t-\rho^1_t\rho^2),\\[.5ex]
&[P^x(\chi),R^x(\alpha)]=Z(\chi\alpha),\quad
[P^y(\rho),R^y(\beta)]=Z(\rho\beta).
\end{split}
\end{gather}

We find megaideals of the algebra~$\mathfrak g$
that will be used for computing the point-symmetry pseudogroup~$G$ of the equation~\eqref{eq:dN}.
The only megaideal that is obvious in view of the above commutation relations is
\[
\mathfrak m_1:=\mathfrak g'=\big\langle
D^t(\tau),P^x(\chi),P^y(\rho),R^x(\alpha),R^y(\beta),Z(\sigma)\big\rangle.
\]
Here and throughout the paper
$\mathfrak z(\mathfrak s)$ and $\mathfrak s'$ denote the center and the derived algebra
of a subalgebra~$\mathfrak s$ of algebra~$\mathfrak g$, respectively,
$\mathfrak s''=(\mathfrak s')'$,
$\mathfrak s'''=(\mathfrak s'')'$ and $\mathfrak s^3:=[\mathfrak s,[\mathfrak s,\mathfrak s]]$.
More generally, $s^{(0)}:=\mathfrak s$ and
the $n$th derived algebra~$\mathfrak s^{(n)}$ of~$\mathfrak s$ is recursively defined
by $\mathfrak s^{(n+1)}=(\mathfrak s^{(n)})'$, $n\in\mathbb N$.

The computation of other megaideals of~$\mathfrak g$ is based on the following assertion.

\begin{lemma}\label{lem:radicalOfg}
The radical $\mathfrak r$ of~$\mathfrak g$ coincides with
$\big\langle D^{\rm s},P^x(\chi),P^y(\rho),R^x(\alpha),R^y(\beta),Z(\sigma)\big\rangle$.
\end{lemma}

\begin{proof}
Following the proof of Lemma~1 in~\cite{malt2021a},
we denote the span from lemma's statement by~$\mathfrak s$.
To conclude that it coincides with the radical~$\mathfrak r$ of~$\mathfrak g$,
we prove that it is the maximal solvable ideal of~$\mathfrak g$.

The commutation relations between the vector fields spanning $\mathfrak g$, see~\eqref{eq:dNCommRelations},
imply that $\mathfrak s$ is an ideal of~$\mathfrak g$.
Since the fourth derived algebra~$\mathfrak s^{(4)}$ of~$\mathfrak s$ is equal to $\{0\}$,
then the ideal $\mathfrak s$ is solvable (of solvability rank four).

Now we show that the solvable ideal $\mathfrak s$ of~$\mathfrak g$ is maximal in~$\mathfrak g$.
Let $\mathfrak s_1$ be an ideal of $\mathfrak g$ properly containing $\mathfrak s$.
This means that at least for one nonvanishing value~$\tau^1$ of the parameter function~$\tau$,
the corresponding vector field $D^t(\tau^1)$ belongs to~$\mathfrak s_1$.
Denote by $I$ an interval in the domain of~$\tau^1$ such that $\tau^1(t)\neq0$ for any $t\in I$.
We restrict all the parameter functions in~$\mathfrak g$ on the interval~$I$.
Since $\mathfrak s_1$ is an ideal of $\mathfrak g$,
the commutator $[D^t(\tau),D^t(\tau^1)]=D^t(\tilde\tau)$ with $\tilde\tau:=\tau\tau^1_t-\tau_t\tau^1$
belongs to~$\mathfrak s_1$ for all~$\tau\in\mathrm C^\infty(I)$.
If the function~$\tau$ runs through~$\mathrm C^\infty(I)$, then,
in view of the existence theorem for first-order linear ordinary differential equations,
the function~$\tilde\tau$ also runs through~$\mathrm C^\infty(I)$.
Therefore, $\mathfrak s_1\supset\langle D^t(\tau)\rangle$ and thus
the $n$th derived algebra~$\mathfrak s_1^{(n)}$ of~$\mathfrak s_1$ contains
$\langle D^t(\tau)\rangle\ne\{0\}$ for any $n\in\mathbb N$ as well,
i.e., the ideal~$\mathfrak s_1$ is not solvable.
Hence the span~$\mathfrak s$ is maximal as a solvable ideal of~$\mathfrak g$.
\end{proof}

We set $\mathfrak m_2:=\mathfrak r$.
In view of properties of megaideals~\cite{bihl2015a,popo2003a},
it is easy to construct several other megaideals of the algebra~$\mathfrak g$,
\begin{gather*}
\mathfrak m_3:=\mathfrak m_2'=\mathfrak m_1\cap\mathfrak m_2
=\big\langle P^x(\chi),P^y(\rho),R^x(\alpha),R^y(\beta),Z(\sigma)\big\rangle,\\
\mathfrak m_4:=\mathfrak m_2''=\big\langle R^x(\alpha),R^y(\beta),Z(\sigma)\big\rangle,\quad
\mathfrak m_5:=(\mathfrak m_3)^3=\mathfrak z(\mathfrak m_3)=\big\{ Z(\sigma)\big\},\\
\mathfrak m_6:=\mathfrak z(\mathfrak m_1)=\big\langle Z(1)\big\rangle.
\end{gather*}
Overall, the algebra~$\mathfrak g$ contains the proper megaideal $\mathfrak m_1=\mathfrak g'$ and
the chain of proper megaideals contained in its radical,
\[
\mathfrak g\varsupsetneq\mathfrak r=:\mathfrak m_2\varsupsetneq\mathfrak m_3
\varsupsetneq\mathfrak m_4\varsupsetneq\mathfrak m_5\varsupsetneq\mathfrak m_6.
\]
Each of them is essential when applying the algebraic method
to construct the point-symmetry pseudogroup of the dispersionless Nizhnik equation~\eqref{eq:dN}
in the sense that it is not the sum of other proper megaideals.
Note that in contrast to the megaideals~$\mathfrak m_j$, $j=1,\dots,6$,
the improper nonzero megaideal, which is the algebra~$\mathfrak g$ itself,
is not essential in this sense since $\mathfrak g=\mathfrak m_1+\mathfrak m_2$.
Among the constructed proper megaideals, only the megaideal~$\mathfrak m_6$
is finite-dimensional and, moreover, it is one-dimensional.
It is clear that within the above elementary consideration,
we cannot answer the question of whether the megaideals~$\mathfrak m_j$, $j=1,\dots,6$,
exhaust the entire set of proper megaideals of the (infinite-dimensional) algebra~$\mathfrak g$.

Since $\mathfrak g_{\rm c}=\mathfrak g_{(1)}\simeq\mathfrak g$,
all the above claims on the structure of the algebra~$\mathfrak g$ are also relevant for the algebra~$\mathfrak g_{\rm c}$
after reformulating them for the corresponding first prolongations,
which are marked by the additional subscript ``(1)''.

\section{Point- and contact-symmetry pseudogroups}\label{sec:PointSymGroup}

\begin{theorem}\label{thm:dNPointSymPseudogroup}
(i) The point-symmetry pseudogroup~$G$ of the dispersionless Nizhnik equation~\eqref{eq:dN}
is generated by the transformations of the form
\begin{gather}\label{eq:dNPointSymForm}
\begin{split}
&\tilde t=T(t),\quad
\tilde x=CT_t^{1/3}x+X^0(t),\quad
\tilde y=CT_t^{1/3}y+Y^0(t),\\
&\tilde u=C^3u-\frac{C^3T_{tt}}{18T_t}(x^3+y^3)
-\frac{C^2}{2T_t^{1/3}}(X^0_tx^2+Y^0_ty^2)+W^1(t)x+W^2(t)y+W^0(t)
\end{split}
\end{gather}
and the transformation~$\mathscr J$: $\tilde t=t$, $\tilde x=y$, $\tilde y=x$, $\tilde u=u$.
Here $T$, $X^0$, $Y^0$, $W^0$, $W^1$ and $W^2$ are arbitrary smooth functions of~$t$ with $T_t\neq0$,
and $C$ is an arbitrary nonzero constant.

(ii) The contact-symmetry pseudogroup~$G_{\rm c}$ of the dispersionless Nizhnik equation~\eqref{eq:dN}
coincides with the first prolongation~$G_{(1)}$ of the pseudogroup~$G$.
\end{theorem}

\begin{proof}
(i) Since the maximal Lie invariance algebra~$\mathfrak g$ of the equation~\eqref{eq:dN} is infinite-dimensional,
we compute the pseudogroup~$G$ using the modification of the megaideal-based method that was suggested in~\cite{malt2021a}.
The application of this method to the equation~\eqref{eq:dN} is based on the following observation.
If a point transformation~$\Phi$ in the space with the coordinates $(t,x,y,u)$,
\[
\Phi\colon\ (\tilde t,\tilde x,\tilde y,\tilde u)=(T,X,Y,U),
\]
where $(T,X,Y,U)$ is a tuple of smooth functions of $(t,x,y,u)$ with nonvanishing Jacobian,
is a point symmetry of the equation~\eqref{eq:dN},
then $\Phi_*\mathfrak m_j\subseteq\mathfrak m_j$, $j=1,\dots,6$.
Here and in what follows
$\Phi_*$ denotes the pushforward of vector fields by~$\Phi$, and $z\in\{x,y\}$.

We choose the following linearly independent elements of~$\mathfrak g$:
\begin{gather*}
\begin{split}&
Q^1:=Z(1),\quad Q^2:=Z(t),\quad Q^{3z}:=R^z(1),\\&
Q^{4z}:=P^z(1),\quad Q^{5z}:=P^z(t),\quad
Q^6:=D^{\rm s},\quad Q^7:=D^t(1),\quad Q^8:=D^t(t).
\end{split}
\end{gather*}
Since
$Q^1\in\mathfrak m_6$, $Q^2\in\mathfrak m_5$,
$Q^{3z}\in\mathfrak m_4$,
$Q^{4z},Q^{5z}\in\mathfrak m_3$,
$Q^6\in\mathfrak m_2$ and
$Q^7,Q^8\in\mathfrak m_1$, then
\begin{gather}\label{eq:dNMainPushforwards}
\begin{split}&
\Phi_*Q^i=\tilde Z(\tilde\sigma^i),\quad i=1,2,
\\&
\Phi_*Q^{iz}=\tilde R^x(\tilde\alpha^{iz})+\tilde R^y(\tilde\beta^{iz})+\tilde Z(\tilde\sigma^{iz}),\quad i=3,
\\&
\Phi_*Q^{iz}=\tilde P^x(\tilde\chi^{iz})+\tilde P^y(\tilde\rho^{iz})
+\tilde R^x(\tilde\alpha^{iz})+\tilde R^y(\tilde\beta^{iz})+\tilde Z(\tilde\sigma^{iz}),\quad i=4,5,
\\&
\Phi_*Q^i=\lambda^i\tilde D^{\rm s}+\tilde P^x(\tilde\chi^i)+\tilde P^y(\tilde\rho^i)
+\tilde R^x(\tilde\alpha^i)+\tilde R^y(\tilde\beta^i)+\tilde Z(\tilde\sigma^i),\quad i=6,
\\&
\Phi_*Q^i=\tilde D^t(\tilde\tau^i)+\tilde P^x(\tilde\chi^i)+\tilde P^y(\tilde\rho^i)
+\tilde R^x(\tilde\alpha^i)+\tilde R^y(\tilde\beta^i)+\tilde Z(\tilde\sigma^i),\quad i=7,8.
\end{split}
\end{gather}
Here
$\lambda^6$ and $\tilde\sigma^1$ are constants,
the other parameters are smooth functions of~$\tilde t$,
and $\tilde\sigma^1\tilde\sigma^2\ne0$.

We will simultaneously present two slightly different proofs,
respectively using elements with $i\in\{1,\dots,6\}$ or with $i\in\{1,\dots,5,7,8\}$.
For each relevant~$i$ and for each $z\in\{x,y\}$ whenever it is relevant,
we expand the corresponding equation from~\eqref{eq:dNMainPushforwards}, split it componentwise and pull the result back by~$\Phi$.
We simplify the obtained constraints, taking into account constraints derived in the same way for preceding values of~$i$
and omitting the constraints satisfied identically in view of other constraints.

Thus, for $i=1,2$, we get
\[
T_u=X_u=Y_u=0,\quad U_u=\tilde\sigma^1,\quad tU_u=\tilde\sigma^2(T).
\]
Since $\tilde\sigma^1\ne0$, this implies $U=U^1u+U^0(t,x,y)$ with constant~$U^1\ne0$
and $t=f(T)$ with $f(\tilde t):=\tilde\sigma^2(\tilde t)/\tilde\sigma^1$.
Differentiating the equality $t=f(T)$ with respect to~$t$ gives $1=f_{\tilde t}(T)T_t$.
Therefore, the derivative $f_{\tilde t}$ does not vanish,
and according to the inverse function theorem, we obtain that $T=T(t)$ with $T_t\ne0$
since the Jacobian of~$\Phi$ does not vanish.

Using the same procedure for $i=3$ results in the equations
\begin{gather}\label{eq:dNSubsysM4}
\begin{split}&
xU^1=\tilde\alpha^{3x}(T)X+\tilde\beta^{3x}(T)Y+\tilde\sigma^{3x}(T),\\&
yU^1=\tilde\alpha^{3y}(T)X+\tilde\beta^{3y}(T)Y+\tilde\sigma^{3y}(T).
\end{split}
\end{gather}
The matrix constituted by the coefficients of~$(X,Y)$ in the system~\eqref{eq:dNSubsysM4} is nondegenerate
since otherwise this system would imply a nonidentity affine constraint for~$(x,y)$ with coefficients depending at most on~$t$.
Solving the system~\eqref{eq:dNSubsysM4} with respect to~$(X,Y)$ leads to the representation
\[
X=X^1(t)x+X^2(t)y+X^0(t),\quad
Y=Y^1(t)x+Y^2(t)y+Y^0(t),
\]
where $X^1Y^2-X^2Y^1\ne0$ due to nonvanishing the Jacobian of~$\Phi$.
We will also need the counterpart of this representation that is solved with respect to~$(x,y)$,
\begin{gather}\label{eq:TempExprForxy}
x=\tilde X^1(t)X+\tilde X^2(t)Y+\tilde X^0(t),\quad
y=\tilde Y^1(t)X+\tilde Y^2(t)Y+\tilde Y^0(t),
\end{gather}
where $\tilde X^1\tilde Y^2-\tilde X^2\tilde Y^1\ne0$ as well, and
\[
\begin{pmatrix}\tilde X^1&\tilde X^2\\ \tilde Y^1&\tilde Y^2\end{pmatrix}=
\begin{pmatrix}X^1& X^2\\ Y^1&Y^2\end{pmatrix}^{-1}
.
\]

Instead of the equations from~\eqref{eq:dNMainPushforwards} with $i=4,5$, we first immediately consider their combinations.
More specifically, for each~$z$ we subtract the equation with $i=4$ multiplied by~$t$ from the equation with $i=5$ and with the same~$z$.
In the obtained equations, we collect the $\tilde u$-components, pull them back by~$\Phi$,
substitute the expressions~\eqref{eq:TempExprForxy} for $(x,y)$ into them
and then collect the coefficients of~$XY$.
The splitting with respect to $(X,Y)$ is allowed here due to the functional independence of~$t$, $X$ and~$Y$.
In view of the inequality $U^1\ne0$, this leads to the equations
$\tilde X^1\tilde X^2=\tilde Y^1\tilde Y^2=0$ or, equivalently, $X^1Y^1=X^2Y^2=0$.
Since $X^1Y^2-X^2Y^1\ne0$, the latter equations imply that either $X^1=Y^2=0$ or $X^2=Y^1=0$.
It is obvious that the transformation~$\mathscr J$: $\tilde t=t$, $\tilde x=y$, $\tilde y=x$, $\tilde u=u$,
which just permutes~$x$ and~$y$, is a point symmetry of the equation~\eqref{eq:dN}.
Composing the corresponding point symmetries of the equation~\eqref{eq:dN} with the transformation~$\mathscr J$
reduces the case $X^1=Y^2=0$ to the case $X^2=Y^1=0$.
Therefore, without loss of generality, we can assume in the rest of the proof that $X^2=Y^1=0$ and thus $X^1Y^2\ne0$.
In the $\tilde u$-components pulled back by~$\Phi$, we can also collect the coefficients of~$x^2$ and of~$y^2$,
which leads to the equations
$U^1
=(X^1)^2(\tilde\chi^{5x}_{\tilde t}(T)-t\tilde\chi^{4x}_{\tilde t}(T))
=(Y^2)^2(\tilde\rho^{5x}_{\tilde t}(T)-t\tilde\rho^{4x}_{\tilde t}(T))$.
Now we proceed with the equations from~\eqref{eq:dNMainPushforwards} with $i=4,5$ in the usual way.
Considering $\tilde x$- and $\tilde y$-components, we derive the equations
$\tilde\chi^{4x}(T)=X^1$, $\tilde\chi^{5x}(T)=tX^1$,
$\tilde\rho^{4y}(T)=Y^2$, $\tilde\rho^{5y}(T)=tY^2$.
Therefore, $(X^1)^3=(Y^2)^3=U^1T_t$, and thus $X^1=Y^2=F:=CT_t^{1/3}\ne0$ and $U^1=C^3$
with constant $C:=(U^1)^{1/3}\ne0$.
The optimal way to obtain the rest of the equations of this step,
\begin{gather*}
U^0_x=-\frac{F_t}{2T_t}(Fx+X^0)^2+\tilde\alpha^{4x}(T)(Fx+X^0)+\tilde\beta^{4x}(T)(Fy+Y^0)+\tilde\sigma^{4x}(T),\\
U^0_y=-\frac{F_t}{2T_t}(Fy+Y^0)^2+\tilde\alpha^{4y}(T)(Fx+X^0)+\tilde\beta^{4y}(T)(Fy+Y^0)+\tilde\sigma^{4y}(T),
\end{gather*}
is to consider the $\tilde u$-components for $i=4$ and $z\in\{x,y\}$.
The compatibility condition of these equations is $U^0_{xy}=U^0_{yx}$, giving $\tilde\beta^{4x}=\tilde\alpha^{4y}$.
Their joint integration implies the representation
\begin{gather*}
U^0=-\frac{F^2F_t}{6T_t}(x^3+y^3)+W^3x^2+W^4xy+W^5y^2+W^1x+W^2y+W^0,
\end{gather*}
where $W^0$, \dots, $W^5$ are smooth functions of~$t$ that are not constrained on this step.

There are two ways for further computations.

The first way is to implement the standard procedure for $i=6$, which gives the equations
$\lambda^6=1$, $\tilde\chi^6(T)=-X^0$, $\tilde\rho^6(T)=-Y^0$,
$W^3-\frac12\tilde\chi^6_{\tilde t}(T)F^2=0$,
$W^4=0$ and
$W^5-\frac12\tilde\rho^6_{\tilde t}(T)F^2=0$.
Hence
\[W^3=-\frac12C^2T_t^{-1/3}X^0_t,\quad W^5=-\frac12C^2T_t^{-1/3}Y^0_t,\]
and we do not need to use the megaideal~$\mathfrak m_1$.

The second way involves certain equations following
from the equations with $i=7,8$ in~\eqref{eq:dNMainPushforwards}.
The corresponding computation is a bit more complicated than the above one
and involves the megaideal~$\mathfrak m_1$ instead of~$\mathfrak m_2$.
Nevertheless, as shown below in Section~\ref{sec:DefSubalgsForPointTrans},
the subalgebra of~$\mathfrak g$ underlying this way has a nicer property
than the analogous subalgebra for the first way.

Thus, up to composing with the transformation~$\mathscr J$,
the transformation~$\Phi$ has the form~\eqref{eq:dNPointSymForm} declared in the statement of the theorem.
It is straightforward to check by the direct substitution that
any point transformation of this form is a point symmetry of the equation~\eqref{eq:dN}.

(ii) To prove the equality $G_{\rm c}=G_{(1)}$,
we apply the same modification of the megaideal-based method,
where the maximal Lie invariance algebra~$\mathfrak g$ is replaced with
the contact invariance algebra~$\mathfrak g_{\rm c}=\mathfrak g_{(1)}$,
and the point transformation~$\Phi$ is replaced with a contact transformation
\begin{subequations}\label{eq:dNGenContactTrans}
\begin{gather}\label{eq:dNGenContactTransA}
\Psi\colon\ (\tilde t,\tilde x,\tilde y,\tilde u,\tilde u_{\tilde t},\tilde u_{\tilde x},\tilde u_{\tilde y})
=(Z^t,Z^x,Z^y,U,U^t,U^x,U^y).
\end{gather}
In~$\Psi$, the tuple on the right-hand side is a tuple of smooth functions of $(t,x,y,u,u_t,u_x,u_y)$
with nonvanishing Jacobian,
which additionally satisfies the contact condition
\begin{gather}\label{eq:dNGenContactTransB}
(Z^\mu_\nu+Z^\mu_uu_\nu)U^\mu=U_\nu+U_uu_\nu,\quad
Z^\mu_{u_\nu}U^\mu=U_{u_\nu}.
\end{gather}
\end{subequations}
Here and in what follows the indices~$\mu$ and~$\nu$
run through the set $\{t,x,y\}$,
and we assume summation for repeated indices.
If the transformation~$\Psi$ is a contact symmetry of the equation~\eqref{eq:dN},
then $\Psi_*\mathfrak m_{j(1)}\subseteq\mathfrak m_{j(1)}$, $j=1,\dots,6$,
where $\Psi_*$ denotes the pushforward of contact vector fields by~$\Psi$.
To the counterpart of the collection of equations~\eqref{eq:dNMainPushforwards} for the contact case,
we apply the procedure that is completely analogous to that described after~\eqref{eq:dNMainPushforwards}.
From the equation with $i=1$, we in particular derive the constraints $Z^\mu_u=0$.
Then the equations with $i=2$, $(i,z)=(3,x)$ and $(i,z)=(3,y)$
imply the constraints $Z^\mu_{u_t}=0$, $Z^\mu_{u_x}=0$ and $Z^\mu_{u_y}=0$, respectively.
In view of the contact condition, this means that $U_{u_\nu}=0$ as well,
and thus the contact transformation~$\Psi$ is the first prolongation
of a point transformation in the space with the coordinates $(t,x,y,u)$.\looseness=1
\end{proof}

\begin{corollary}\label{cor:dNDiscrSyms}
The identity component~$G_{\rm id}$ of the point-symmetry pseudogroup~$G$ of the dispersionless Nizhnik equation~\eqref{eq:dN}
consists of the transformations of the form~\eqref{eq:dNPointSymForm} with $T_t>0$ and $C>0$.
A complete list of discrete point symmetry transformations of the equation~\eqref{eq:dN}
that are independent up to composing with each other and with transformations from~$G_{\rm id}$
is exhausted by three commuting involutions, which can be chosen to be
the permutation~$\mathscr J$ of the variables~$x$ and~$y$, $(\tilde t,\tilde x,\tilde y,\tilde u)=(t,y,x,u)$,
and two transformations~$\mathscr I^{\rm i}$ and~$\mathscr I^{\rm s}$ alternating the signs of $(t,x,y)$ and of $(x,y,u)$, respectively,
$\mathscr I^{\rm i}\colon(\tilde t,\tilde x,\tilde y,\tilde u)=(-t,-x,-y,u)$,
$\mathscr I^{\rm s}\colon(\tilde t,\tilde x,\tilde y,\tilde u)=(t,-x,-y,-u)$.
\end{corollary}

Therefore, the quotient group~$G/G_{\rm id}$ of the pseudogroup~$G$
with respect to its identity component~$G_{\rm id}$ is isomorphic to the group $\mathbb Z_2\times\mathbb Z_2\times\mathbb Z_2$.

\begin{remark}\label{rem:dNinvolutions}
In Corollary~\ref{cor:dNDiscrSyms} and analogous Corollaries~\ref{cor:dNLaxPairDiscrSyms} and~\ref{cor:dNSystemDiscrSyms} below,
we assume that each of the listed discrete transformations is defined on the entire corresponding underlying space
and absorbs its restrictions.
The claims on the structure of the related discrete groups after the indicated corollaries are rigorous only
under this assumption.
The same assumption should be imposed on the elements of~$\mathfrak H$
in the proof of Theorem~\ref{thm:dNSystemPointSymPseudogroup} for~$\mathfrak H$ to be a group.
\end{remark}

\section{Defining subalgebras for point transformations}\label{sec:DefSubalgsForPointTrans}

In the course of applying the megaideal-based method
in the proof of Theorem~\ref{thm:dNPointSymPseudogroup},
we expand the basic condition $\Phi_*Q\in\mathfrak m$,
where $\mathfrak m$ is the minimal megaideal of~$\mathfrak g$ containing the vector field~$Q$,
only for 11 (linearly independent) vector fields from the algebra~$\mathfrak g$,
which is infinite-dimensional.
These vector fields span an 11-dimensional subalgebra~$\mathfrak s$ of~$\mathfrak g$.
In fact, we separately consider two subalgebras of~$\mathfrak s$,
\begin{gather*}
\mathfrak s_1=\langle Z(1),Z(t),R^x(1),R^y(1),P^x(1),P^y(1),P^x(t),P^y(t),D^{\rm s}\rangle,\\
\mathfrak s_2=\langle Z(1),Z(t),R^x(1),R^y(1),P^x(1),P^y(1),P^x(t),P^y(t),D^t(1),D^t(t)\rangle.
\end{gather*}
Moreover, in the other example of applying the same modification of the megaideal-based method in~\cite{malt2021a},
which was computing the point-symmetry group of the Boiti--Leon--Pempinelli system,
the selected linearly independent vector fields also span a subalgebra of the corresponding maximal Lie invariance algebra.
Nevertheless, it is still not well understood how common this phenomenon is.
That subalgebra has the following interesting property:

\begin{definition}\label{def:DefiningSubalg}
We call a proper subalgebra~$\mathfrak s$ of a Lie algebra~$\mathfrak a$ of vector fields
a \emph{subalgebra defining the diffeomorphisms that stabilize~$\mathfrak a$}
if the conditions $\Phi_*\mathfrak a\subseteq\mathfrak a$ and $\Phi_*\mathfrak s\subseteq\mathfrak a$
for local diffeomorphisms~$\Phi$ in the underlying space are equivalent.
\end{definition}

The implication $\Phi_*\mathfrak a\subseteq\mathfrak a\Rightarrow\Phi_*\mathfrak s\subseteq\mathfrak a$ is obvious,
whereas the inverse implication does not hold in general, and its verification requires nontrivial computations.

For a better understanding of the general foundations of the algebraic method in question,
it is instructive to check whether the subalgebras~$\mathfrak s_1$ and~$\mathfrak s_2$ are of the kind
introduced in Definition~\ref{def:DefiningSubalg}.

\begin{theorem}\label{thm:dNDefSubalgs}
The subalgebra~$\mathfrak s_2$ of the algebra~$\mathfrak g$ defines the diffeomorphisms that stabilize~$\mathfrak g$,
whereas the subalgebra~$\mathfrak s_1$
and even the subalgebra~$\bar{\mathfrak s}_1:=\mathfrak s_1+\langle D^t(1)\rangle$ does not have this property.
\end{theorem}

\begin{proof}
We follow the proof of Theorem~\ref{thm:dNPointSymPseudogroup}
and use the same numeration of the selected elements of the algebra~$\mathfrak g$,
but for each basis element~$Q$ of the subalgebra~$\mathfrak s_1$ we employ the condition $\Phi_*Q\in\mathfrak g$
instead of the condition $\Phi_*Q\in\mathfrak m$,
where $\mathfrak m$ is the minimal megaideal of~$\mathfrak g$ containing the vector field~$Q$.
In other words, we replace the equations~\eqref{eq:dNMainPushforwards} with the equations
\begin{gather}\label{eq:dNMainPushforwardsMod}
\Phi_*Q^\kappa=\lambda^\kappa\tilde D^{\rm s}+\tilde D^t(\tilde\tau^\kappa)+\tilde P^x(\tilde\chi^\kappa)+\tilde P^y(\tilde\rho^\kappa)
+\tilde R^x(\tilde\alpha^\kappa)+\tilde R^y(\tilde\beta^\kappa)+\tilde Z(\tilde\sigma^\kappa),
\end{gather}
where
$\lambda^\kappa$ are constants,
$\tilde\tau^\kappa$, $\tilde\chi^\kappa$, $\tilde\rho^\kappa$, $\tilde\alpha^\kappa$, $\tilde\beta^\kappa$ and~$\tilde\sigma^\kappa$
are smooth functions of~$\tilde t$, and
the index~$\kappa$ runs the set $\big\{1,2,3z,4z,5z,6,7,8\mid z\in\{x,y\}\big\}$.

Collecting $\tilde t$-components in the equations with $\kappa=1,2$,
we derive the equations $T_u=\tilde\tau^1(T)$ and $tT_u=\tilde\tau^2(T)$.
Suppose that $T_u\ne0$, and thus $\tilde\tau^1\tilde\tau^2\ne0$.
Recombining the above equations leads to the equation
$t=f(T)$ with $f(\tilde t):=\tilde\tau^2(\tilde t)/\tilde\tau^1(\tilde t)$.
Differentiating it with respect to~$t$ gives $1=f_{\tilde t}(T)T_t$.
Therefore, the derivative $f_{\tilde t}$ does not vanish,
and according to the inverse function theorem, we obtain that $T=T(t)$,
which contradicts the supposition $T_u\ne0$.
Therefore, $T_u=0$ and also $\tilde\tau^1=\tilde\tau^2=0$.

Collecting $\tilde x$- and $\tilde y$-components in the same equations with $\kappa=1,2$
leads to the equations
$ X_u=\lambda^1X+\tilde\chi^1(T)$,
$tX_u=\lambda^2X+\tilde\chi^2(T)$,
$ Y_u=\lambda^1Y+\tilde\rho^1(T)$ and
$tY_u=\lambda^2Y+\tilde\rho^2(T)$,
which can be combined to
$(\lambda^1t-\lambda^2)X+t\tilde\chi^1(T)-\tilde\chi^2(T)=0$ and
$(\lambda^1t-\lambda^2)Y+t\tilde\rho^1(T)-\tilde\rho^2(T)=0$.
Suppose that $(X_u,Y_u)\ne(0,0)$.
Then we can split at least one of the last two equations with respect to
$X$ or~$Y$, respectively.
As a result, we obtain the equation $\lambda^1t-\lambda^2=0$,
which splits further with respect to~$t$ to $\lambda^1=\lambda^2=0$.
Therefore, we also have
$t\tilde\chi^1(T)=\tilde\chi^2(T)$ and
$t\tilde\rho^1(T)=\tilde\rho^2(T)$.
Moreover, $(\tilde\chi^1\tilde\chi^2,\tilde\rho^1\tilde\rho^2)\ne(0,0)$
due to the supposition $(X_u,Y_u)\ne(0,0)$.
Following the consideration of $t$-components, we again derive an equation of the form
$t=f(T)$ with $f_{\tilde t}\ne0$ and obtain in view of the inverse function theorem that $T=T(t)$.
Then we collect $\tilde u$-components in the same equations with $\kappa=1,2$
and derive the equations
\begin{gather*}
 U_u=-\frac12\tilde\chi^1_{\tilde t}(T)X^2-\frac12\tilde\rho^1_{\tilde t}(T)Y^2
+\tilde\alpha^1(T)X+\tilde\beta^1(T)Y+\tilde\sigma^1(T),
\\
tU_u=-\frac12\tilde\chi^2_{\tilde t}(T)X^2-\frac12\tilde\rho^2_{\tilde t}(T)Y^2
+\tilde\alpha^2(T)X+\tilde\beta^2(T)Y+\tilde\sigma^2(T).
\end{gather*}
We subtract the second equation from the first one multiplied by~$t$.
Since $t$, $X$ and~$Y$ are functionally independent,
the equation obtained in this way can be split with respect to~$(X,Y)$,
which in particular results in the equations
$t\tilde\chi^1_{\tilde t}(T)=\tilde\chi^2_{\tilde t}(T)$,
$t\tilde\rho^1_{\tilde t}(T)=\tilde\rho^2_{\tilde t}(T)$.
Pairwise differential consequences of these equations jointly with the equations
$t\tilde\chi^1(T)=\tilde\chi^2(T)$ and
$t\tilde\rho^1(T)=\tilde\rho^2(T)$ are the equations
$\tilde\chi^1=\tilde\chi^2=0$ and $\tilde\rho^1=\tilde\rho^2=0$, respectively.
Therefore, we have the equations $X_u=Y_u=0$,
which contradict the supposition $(X_u,Y_u)\ne(0,0)$.
This is why in fact $X_u=Y_u=0$ as well as
$\lambda^1=\lambda^2=0$,
$\tilde\chi^1=\tilde\chi^2=0$, $\tilde\rho^1=\tilde\rho^2=0$ and $U_u\ne0$.

Under the derived constraints,
the only essential equation that is obtained via collecting $\tilde u$-components
in the equations with $\kappa=1,2$ is
$U_u=\tilde\alpha^1(T)X+\tilde\beta^1(T)Y+\tilde\sigma^1(T)$.

We temporarily jump to the equations with $\kappa=4z,5z$, $z\in\{x,y\}$,
where we only collect $\tilde t$-components on this step, obtaining
$T_z=\tilde\tau^{4z}(T)$, $tT_z=\tilde\tau^{5z}(T)$, and thus $t\tilde\tau^{4z}(T)=\tilde\tau^{5z}(T)$.
Supposing that $T_z\ne0$ for some $z\in\{x,y\}$, we then have $\tilde\tau^{4z}(T)\ne0$ and
$t=f(T)$ with $f=\tilde\tau^{5z}/\tilde\tau^{4z}$.
Using the same arguments as at the beginning of the proof,
we obtain that the function~$T$ depends only on~$t$, which contradicts the supposition $T_z\ne0$.
Hence $T_x=T_y=0$, i.e.,
$T$ is nevertheless a function of~$t$ only, $T=T(t)$ with $T_t\ne0$,
and also $\tilde\tau^{4z}=\tilde\tau^{5z}=0$.

We return to the equations with $\kappa=3z$, $z\in\{x,y\}$,
which we also consider simultaneously.
We successively collect $\tilde t$-, $\tilde x$- and $\tilde y$-components
and split the obtained equations with respect to~$X$ and~$Y$
since the functions~$T$, $X$ and~$Y$ are functionally independent.
This gives the constraints $\tilde\tau^{3z}=\tilde\chi^{3z}=\tilde\rho^{3z}=0$, $\lambda^{3z}=0$.
Then collecting of $\tilde u$-components leads to the constraints
$zU_u=\tilde\alpha^{3z}(T)X+\tilde\beta^{3z}(T)Y+\tilde\sigma^{3z}(T)$.
In view of the above expression for~$U_u$,
this means that $x$ and~$y$ can be represented as linear fractional functions of $(X,Y)$
with coefficients depending on~$T$.
Since the inverse~$\Phi^{-1}$ belongs to the pseudogroup~$G$ as the transformation~$\Phi$ does,
we can permit $(t,x,y)$ and $(T,X,Y)$ in the last claim.
In other words, $X$ and~$Y$ are linear fractional functions of $(x,y)$
with coefficients depending on~$t$, $X={\rm N}^X/{\rm D}$ and $Y={\rm N}^Y/{\rm D}$,
where the numerators and the denominator respectively are
${\rm N}^X=X^1(t)x+X^2(t)y+X^0(t)$, ${\rm N}^Y=Y^1(t)x+Y^2(t)y+Y^0(t)$ and ${\rm D}=K^1(t)x+K^2(t)y+K^0(t)$
for some smooth functions $X^0$, $X^1$, $X^2$, $Y^0$, $Y^1$, $Y^2$, $K^0$, $K^1$ and $K^2$ of~$t$.

Now we consider the equations with $\kappa=4z$, $z\in\{x,y\}$.
The equations $X_z=\lambda^{4z}X+\tilde\chi^{4z}(T)$ and $Y_z=\lambda^{4z}Y+\tilde\rho^{4z}(T)$
which are obtained by successively collecting $\tilde x$- and $\tilde y$-components,
reduce to
\begin{gather}\label{eq:Mu4z}
\begin{split}
&X^1{\rm D}-K^1{\rm N}^X=\lambda^{4x}{\rm N}^X{\rm D}+\tilde\chi^{4x}(T){\rm D}^2,\quad
 X^2{\rm D}-K^2{\rm N}^X=\lambda^{4y}{\rm N}^X{\rm D}+\tilde\chi^{4y}(T){\rm D}^2,\\
&Y^1{\rm D}-K^1{\rm N}^Y=\lambda^{4x}{\rm N}^Y{\rm D}+\tilde\rho^{4x}(T){\rm D}^2,\quad
 Y^2{\rm D}-K^2{\rm N}^Y=\lambda^{4y}{\rm N}^Y{\rm D}+\tilde\rho^{4y}(T){\rm D}^2.
\end{split}
\end{gather}
Suppose that $(K^1,K^2)\ne(0,0)$.
Then $X^1K^2-X^2K^1\ne0$ or $Y^1K^2-Y^2K^1\ne0$
since otherwise the Jacobian of the functions~$T$, $X$ and~$Y$ is zero.
Recall that the point transformation~$\mathscr J$: $\tilde t=t$, $\tilde x=y$, $\tilde y=x$, $\tilde u=u$,
which just permutes~$x$ and~$y$, is an obvious point symmetry of the equation~\eqref{eq:dN}.
This is why we can assume without loss of generality that $X^1K^2-X^2K^1\ne0$.
Hence the Jacobian of the functions~${\rm N}^X$ and~${\rm D}$ with respect to $(x,y)$ is nonzero,
and we can split the first two equations in~\eqref{eq:Mu4z} with respect to $({\rm N}^X,{\rm D})$.
As a result, we in particular derive the constraints $X^1=X^2=0$,
which contradict the inequality $X^1K^2-X^2K^1\ne0$.
Therefore, $K^1=K^2=0$, i.e., the functions~$X$ and~$Y$ are affine in $(x,y)$ with coefficients depending on~$t$.
Re-denoting $X^k/K^0$ by~$X^k$ and $Y^k/K^0$ by~$Y^k$, $k=0,1,2$,
we can completely follow the part with $i=4,5$ in item (i) of the proof of Theorem~\ref{thm:dNPointSymPseudogroup}.

The computation for $\kappa=6$ and further is again different.

Taking only the $\tilde t$-components and
the coefficients of~$x$ in the $\tilde x$-components or, equivalently,
the coefficients of~$y$ in the $\tilde y$-components in~\eqref{eq:dNMainPushforwardsMod} with $\kappa=7$ and with $\kappa=8$,
we respectively obtain
$T_t=\tilde\tau^7(T)$, \smash{$F_t/F=\frac13T_{tt}/T_t+\lambda^7$} and
$tT_t=\tilde\tau^8(T)$, \smash{$tF_t/F=\frac13tT_{tt}/T_t+\lambda^8$}.
Combining the second and fourth equations to exclude~$F_t/F$ gives $t\lambda^7=\lambda^8$, i.e., $\lambda^7=\lambda^8=0$,
and thus we can complete the proof for the subalgebra~$\mathfrak s_2$
as in the second way in item (i) on the proof of Theorem~\ref{thm:dNPointSymPseudogroup}.
Therefore, the condition $\Phi_*\mathfrak s_2\subseteq\mathfrak g$ implies $\Phi_*\mathfrak g\subseteq\mathfrak g$.
In other words, the subalgebra~$\mathfrak s_2$ defines the diffeomorphisms that stabilize~$\mathfrak g$.

If we use the equations~\eqref{eq:dNMainPushforwardsMod} with $\kappa=6,7$ instead of those with $\kappa=7,8$,
then we again obtain the equations
\[
\frac{F_t}F=\frac{T_{tt}}{3T_t}+\lambda^7,\quad
W^3=-\frac{F^2}{2T_{t}}X^0_t,\quad W^4=0,\quad
W^5=-\frac{F^2}{2T_{t}}Y^0_t
\]
for transformation parameters, and only these equations and their differential consequences.
Here the parameter~$\lambda^7$ is an arbitrary constant,
and thus the set of point transformations~$\Phi$ satisfying the condition $\Phi_*\mathfrak s_1\subseteq\mathfrak g$
properly contains the group~$G$,
which coincides, in view of Theorem~\ref{thm:dNPointSymPseudogroup},
with the set of point transformations~$\Phi$ satisfying the condition $\Phi_*\mathfrak g\subseteq\mathfrak g$.
Therefore, the subalgebra~$\bar{\mathfrak s}_1$
does not define completely the diffeomorphisms that stabilize~$\mathfrak g$.
Then the subalgebra~$\mathfrak s_1$ as that contained in~$\bar{\mathfrak s}_1$
all the more has the same property.
\end{proof}

\section{Defining subalgebras for contact transformations}\label{sec:DefSubalgsForContactTrans}

Subalgebras defining the diffeomorphisms that stabilize the entire corresponding algebras
can also be considered for algebras of contact vector fields and local contact diffeomorphisms.
The transition from the point case to the contact one complicates the problem
due to extending the space coordinatized by the independent and dependent variables
with the first-order jet variables and thus essentially increasing the total number of coordinates.

Theorem~\ref{thm:dNDefSubalgs} implies
that the first prolongation~$\mathfrak s_{1(1)}$
of the subalgebra~$\mathfrak s_1$ of~$\mathfrak g$
does not define local contact diffeomorphisms that stabilize~$\mathfrak g_{(1)}$
since the subalgebra~$\mathfrak s_1$ itself does not define
local diffeomorphisms that stabilize~$\mathfrak g$.
It is not clear whether the first prolongation~$\mathfrak s_{2(1)}$
of the subalgebra~$\mathfrak s_2$ of~$\mathfrak g$ differs from~$\mathfrak s_{1(1)}$
in the sense of defining local contact diffeomorphisms that stabilize~$\mathfrak g_{(1)}$.
To answer this question it is necessary to integrate
cumbersome parameterized nonlinear overdetermined systems of differential equations,
and its solution requires more sophisticated techniques
than those used in the proofs of Theorems~\ref{thm:dNPointSymPseudogroup} and~\ref{thm:dNDefSubalgs}.
The latter techniques are still efficient only if we extend the algebra to be tested.

\begin{theorem}\label{thm:dNContDefSubalgs}
A contact transformation~$\Psi$ with the basic space $\mathbb R^3_{t,x,y}\times\mathbb R_u$
satisfies the condition $\Psi_*\mathfrak s_{3(1)}\subseteq\mathfrak g_{(1)}$ for the subalgebra
\[
\mathfrak s_3=\langle Z(1),Z(t),Z(t^2),R^x(1),R^y(1),R^x(t),R^y(t)\rangle
\]
of the algebra~$\mathfrak g$
only if it is the first prolongation of a point transformation in the above space.
\end{theorem}

\begin{proof}
Suppose that a contact transformation~$\Psi$ with the basic space $\mathbb R^3_{t,x,y}\times\mathbb R_u$,
which is of the general form~\eqref{eq:dNGenContactTrans},
satisfies the condition $\Psi_*\mathfrak s_{3(1)}\subseteq\mathfrak g_{(1)}$.
For convenience, we re-denote the basis elements of~$\mathfrak s_3$ as
\begin{gather*}
\mathscr Q^1:=Z(1),\quad
\mathscr Q^2:=Z(t),\quad
\mathscr Q^3:=Z(t^2),\\
\mathscr Q^4:=R^x(1),\quad
\mathscr Q^5:=R^y(1),\quad
\mathscr Q^6:=R^x(t),\quad
\mathscr Q^7:=R^y(t).
\end{gather*}
Then we expand the condition $\Psi_*\mathfrak s_{3(1)}\subseteq\mathfrak g_{(1)}$ to
\begin{gather}\label{eq:dNContactPushforwards}
\Psi_*\mathscr Q^i_{(1)}=\lambda^i\tilde D^{\rm s}_{(1)}+\tilde D^t_{(1)}(\tilde\tau^i)
+\tilde P^x_{(1)}(\tilde\chi^i)+\tilde P^y_{(1)}(\tilde\rho^i)
+\tilde R^x_{(1)}(\tilde\alpha^i)+\tilde R^y_{(1)}(\tilde\beta^i)+\tilde Z_{(1)}(\tilde\sigma^i),
\end{gather}
where
$\lambda^i$ are constants,
$\tilde\tau^i$, $\tilde\chi^i$, $\tilde\rho^i$, $\tilde\alpha^i$, $\tilde\beta^i$ and~$\tilde\sigma^i$
are smooth functions of~$\tilde t$, and
the index~$i$ runs from~1 to~7.

Collecting $t$-components in the equations~\eqref{eq:dNContactPushforwards} with $i=1,2,3$,
we derive the equations
\[
T_u=\tilde\tau^1(T),\quad
tT_u+T_{u_t}=\tilde\tau^2(T),\quad
t^2T_u+2tT_{u_t}=\tilde\tau^3(T),
\]
whose algebraic consequence is the equation $\tilde\tau^1(T)t^2-2\tilde\tau^2(T)t+\tilde\tau^3(T)=0$.
Suppose that $(\tilde\tau^1,\tilde\tau^2)\ne(0,0)$.
Then the last equation implies that $t=f(T)$, and, similarly to the proof of Theorem~\ref{thm:dNDefSubalgs},
we successively have $T=T(t)$, $\tilde\tau^1=0$ and $\tilde\tau^2=0$, which contradicts the supposition.
Hence $\tilde\tau^1=\tilde\tau^2=0$ and $T_u=T_{u_t}=0$, and thus $\tilde\tau^3=0$ as well.

The next step is to collect  $x$- and $y$-components in the same equations with $i=1,2,3$.
It results in the equations
\begin{gather*}
X_u=\lambda^1X+\tilde\chi^1(T),\quad
tX_u+X_{u_t}=\lambda^2X+\tilde\chi^2(T),\quad
t^2X_u+2tX_{u_t}=\lambda^3X+\tilde\chi^3(T),
\\
Y_u=\lambda^1Y+\tilde\rho^1(T),\quad
tY_u+Y_{u_t}=\lambda^2Y+\tilde\rho^2(T),\quad
t^2Y_u+2tY_{u_t}=\lambda^3Y+\tilde\rho^3(T).
\end{gather*}
We separately combine the equations in each row to exclude derivatives of~$X$ and~$Y$,
\begin{gather*}
(t^2\lambda^1+2t\lambda^2-\lambda^3)X+t^2\tilde\chi^1(T)+2t\tilde\chi^2(T)-\tilde\chi^3(T)=0,\\
(t^2\lambda^1+2t\lambda^2-\lambda^3)Y+t^2\tilde\rho^1(T)+2t\tilde\rho^2(T)-\tilde\rho^3(T)=0.
\end{gather*}
Suppose that $(X_u,X_{u_t},Y_u,Y_{u_t})\ne(0,0,0,0)$.
Then, we can split the last system with respect to~$X$ and~$Y$,
which leads to the equations
$t^2\lambda^1+2t\lambda^2-\lambda^3=0$,
\begin{subequations}\label{eq:dNS3SystemForChiRho}
\begin{gather}\label{eq:dNS3SystemForChiRhoA}
t^2\tilde\chi^1(T)+2t\tilde\chi^2(T)-\tilde\chi^3(T)=0,\quad
t^2\tilde\rho^1(T)+2t\tilde\rho^2(T)-\tilde\rho^3(T)=0.
\end{gather}
The first equation means that $\lambda^1=\lambda^2=\lambda^3=0$,
and hence $(\tilde\chi^1,\tilde\chi^2,\tilde\rho^1,\tilde\rho^2)\ne(0,0,0,0)$.
Following the above consideration of $t$-components, we obtain that
the function~$T$ depends only on~$t$, $T=T(t)$.
We continue the analysis of the equations~\eqref{eq:dNContactPushforwards} with $i=1,2,3$, collecting $u$-components.
This gives the equations
\begin{gather*}
U_u=
-\frac12\tilde\chi^1_{\tilde t}(T)X^2
-\frac12\tilde\rho^1_{\tilde t}(T)Y^2
+\tilde\alpha^1(T)X+\tilde\beta^1(T)Y+\tilde\sigma^1(T),
\\
tU_u+U_{u_t}=
-\frac12\tilde\chi^2_{\tilde t}(T)X^2
-\frac12\tilde\rho^2_{\tilde t}(T)Y^2
+\tilde\alpha^2(T)X+\tilde\beta^2(T)Y+\tilde\sigma^2(T),
\\
t^2U_u+2tU_{u_t}=
-\frac12\tilde\chi^3_{\tilde t}(T)X^2
-\frac12\tilde\rho^3_{\tilde t}(T)Y^2
+\tilde\alpha^3(T)X+\tilde\beta^3(T)Y+\tilde\sigma^3(T).
\end{gather*}
We linearly combine the first, the second and the third equations with coefficients~$t^2$, $-2t$ and~$1$, respectively.
We can split the obtained algebraic consequence of these equations with respect to~$X$ and~$Y$
since $T$, $X$ and $Y$ are functionally independent, $T$ depends on~$t$ only
and thus $t$, $X$ and $Y$ are functionally independent.
As a result, we derive the equations
\begin{gather}\label{eq:dNS3SystemForChiRhoB}
t^2\tilde\chi^1_{\tilde t}(T)+2t\tilde\chi^2_{\tilde t}(T)-\tilde\chi^3_{\tilde t}(T)=0,\quad
t^2\tilde\rho^1_{\tilde t}(T)+2t\tilde\rho^2_{\tilde t}(T)-\tilde\rho^3_{\tilde t}(T)=0.
\end{gather}
In view of them, the analogous consideration of $u_t$-components in the equations with $i=1,2,3$
gives the equations
\begin{gather}\label{eq:dNS3SystemForChiRhoC}
t^2\tilde\chi^1_{\tilde {tt}}(T)+2t\tilde\chi^2_{\tilde {tt}}(T)-\tilde\chi^3_{\tilde {tt}}(T)=0,\quad
t^2\tilde\rho^1_{\tilde {tt}}(T)+2t\tilde\rho^2_{\tilde {tt}}(T)-\tilde\rho^3_{\tilde {tt}}(T)=0.
\end{gather}
\end{subequations}
We construct, separately for the equations with respect to~$\chi$ and for the equations with respect to~$\rho$,
the differential consequences of the system~\eqref{eq:dNS3SystemForChiRho}
that have the structures
$\p_t\eqref{eq:dNS3SystemForChiRhoA}-T_t\eqref{eq:dNS3SystemForChiRhoB}$ and
$\p_t^{\,2}\eqref{eq:dNS3SystemForChiRhoA}-(2T_t\p_t+T_{tt})\eqref{eq:dNS3SystemForChiRhoB}+T_t^{\,2}\eqref{eq:dNS3SystemForChiRhoC}$.
After the additional division by~2, these differential consequences take the form
$t\tilde\chi^1(T)+\tilde\chi^2(T)=0$, $t\tilde\rho^1(T)+\tilde\rho^2(T)=0$,
$\tilde\chi^1(T)=0$ and $\tilde\rho^1(T)=0$
and implies, jointly with~\eqref{eq:dNS3SystemForChiRhoA} that
$\tilde\chi^1=\tilde\chi^2=\tilde\chi^3=0$ and $\tilde\rho^1=\tilde\rho^2=\tilde\rho^3=0$,
which contradicts the supposition $(X_u,X_{u_t},Y_u,Y_{u_t})\ne(0,0,0,0)$.
Hence $X_u=Y_u=X_{u_t}=Y_{u_t}=0$.

The next step is to collect $t$-components in the equations with $i=4,5$.
It gives the equations $T_{u_z}=\tilde\tau^{4z}(T)$ and $tT_{u_z}=\tilde\tau^{5z}(T)$ with $z\in\{x,y\}$.
Suppose that $T_{u_z}\ne0$.
Then, similarly to the beginning of the proof of Theorem~\ref{thm:dNDefSubalgs}
we again derive that $T=T(t)$ and thus $T_{u_z}=0$, which contradicts the supposition.
This implies $T_{u_z}=0$ and $\tilde\tau^{4z}=\tilde\tau^{5z}=0$ as well.

Collecting $x$- and $y$-components in the same equations with $i=4,5$ leads to the equations
\begin{gather*}
X_{u_z}=\lambda^{4z}X+\tilde\chi^{4z}(T),\quad tX_{u_z}=\lambda^{5z}X+\tilde\chi^{5z}(T),\\
Y_{u_z}=\lambda^{4z}Y+\tilde\rho^{4z}(T),\quad tY_{u_z}=\lambda^{5z}Y+\tilde\rho^{5z}(T).
\end{gather*}
They are combined to
$(t\lambda^{4z}-\lambda^{5z})X+t\tilde\chi^{4z}-\tilde\chi^{5z}=0$ and
$(t\lambda^{4z}-\lambda^{5z})Y+t\tilde\rho^{4z}-\tilde\rho^{5z}=0$.
Suppose that $(X_{u_x},Y_{u_x},X_{u_y},Y_{u_y})\ne(0,0,0,0)$.
Then we can successively split these combinations with respect to~$X$ and~$Y$
and in addition split the coefficients of~$X$ and~$Y$ with respect to~$t$,
which gives $\lambda^{4z}=\lambda^{5z}=0$,
$t\tilde\chi^{4z}=\tilde\chi^{5z}$ and
$t\tilde\rho^{4z}=\tilde\rho^{5z}$ with $z\in\{x,y\}$.
After collecting $u$-components in the equations with $i=4,5$, we have
\begin{gather*}
zU_u+U_{u_z}=-\frac12\tilde\chi^{4z}_{\tilde t}(T)X^2
-\frac12\tilde\rho^{4z}_{\tilde t}(T)Y^2
+\tilde\alpha^{4z}(T)X+\tilde\beta^{4z}(T)Y+\tilde\sigma^{4z}(T),
\\
t(zU_u+U_{u_z})=-\frac12\tilde\chi^{5z}_{\tilde t}(T)X^2
-\frac12\tilde\rho^{5z}_{\tilde t}(T)Y^2
+\tilde\alpha^{5z}(T)X+\tilde\beta^{5z}(T)Y+\tilde\sigma^{5z}(T).
\end{gather*}
We combine the last equations, subtracting the first equation multiplied by~$t$ from the second one.
Splitting the combination with respect to~$X$ and $Y$ gives
$t\tilde\chi^{4z}_{\tilde t}-\tilde\chi^{5z}_{\tilde t}=0$ and
$t\tilde\rho^{4z}_{\tilde t}-\tilde\rho^{5z}_{\tilde t}=0$.
Differential consequences of the derived system for~$\chi^{4z}$ and~$\rho^{5z}$
are the equations $\tilde\chi^{4z}=\tilde\chi^{5z}=\tilde\rho^{4z}=\tilde\rho^{5z}=0$,
which obviously imply $X_{u_x}=Y_{u_x}=X_{u_y}=Y_{u_y}=0$.
The obtained contradiction with the supposition $(X_{u_x},Y_{u_x},X_{u_y},Y_{u_y})\ne(0,0,0,0)$
means that we ultimately have $X_{u_x}=Y_{u_x}=X_{u_y}=Y_{u_y}=0$.

Due to the independence of $(T,X,Y)$ on $u_t,u_x,u_y$,
it follows from the contact condition~\eqref{eq:dNGenContactTransB} that $U_{u_t}=U_{u_x}=U_{u_y}=0$ as well.
Therefore, the contact transformation $\Psi$ is the first prolongation of a point transformation
in the basic space $\mathbb R^3_{t,x,y}\times\mathbb R_u$.
\end{proof}

\begin{corollary}\label{cor:dNContDefSubalgs}
The first prolongation of the subalgebra~$\mathfrak s_2+\mathfrak s_3$,
which is a subalgebra of the algebra~$\mathfrak g_{\rm c}=\mathfrak g_{(1)}$,
defines the diffeomorphisms of the corresponding first-order jet space that stabilize~$\mathfrak g_{\rm c}$.
\end{corollary}

\begin{proof}
If a contact transformation~$\Psi$ with the basic space $\mathbb R^3_{t,x,y}\times\mathbb R_u$
satisfies the condition $\Psi_*(\mathfrak s_2+\mathfrak s_3)_{(1)}\subseteq\mathfrak g_{(1)}$,
then it satisfies the weaker condition $\Psi_*\mathfrak s_{3(1)}\subseteq\mathfrak g_{(1)}$.
In view of Theorem~\ref{thm:dNContDefSubalgs}, this implies
that the contact transformation~$\Psi$ is the first prolongation of a point transformation~$\Phi$
in the basic space $\mathbb R^3_{t,x,y}\times\mathbb R_u$, $\Psi=\Phi_{(1)}$
and the condition $\Psi_*(\mathfrak s_2+\mathfrak s_3)_{(1)}\subseteq\mathfrak g_{(1)}$
reduces to the condition $\Phi_*(\mathfrak s_2+\mathfrak s_3)\subseteq\mathfrak g$.
In particular, $\Phi_*\mathfrak s_2\subseteq\mathfrak g$.
According to Theorem~\ref{thm:dNDefSubalgs}, we have $\Phi_*\mathfrak g\subseteq\mathfrak g$.
The first prolongation of the last condition gives the required property of~$\Psi$,
$\Psi_*\mathfrak g_{(1)}\subseteq\mathfrak g_{(1)}$.
\end{proof}

\section{Point-symmetry pseudogroup of nonlinear Lax representation}\label{sec:PointSymGroupOfLaxPair}

A nonlinear Lax representation of the dispersionless Nizhnik equation~\eqref{eq:dN},%
\footnote{%
The corresponding linear nonisospectral Lax representation is
$\chi_t=(p^2+p^{-4}u_{xy}^{\,\,3}+u_{xx}+p^{-2}u_{xy}u_{yy})\chi_x
-(pu_{xxx}-p^{-1}(u_{xy}u_{yy})_x-p^{-3}u_{xy}^{\,\,2}u_{xxy})\chi_p$,
$\chi_y=p^{-2}u_{xy}\chi_x+p^{-1}u_{xxy}\chi_p$,
where $p$ is a variable spectral parameter, $\chi=\chi(t,x,y,p)$ and  $u=u(t,x,y)$.
See, e.g., \cite[p.~360]{serg2018a} and references therein
for linear nonisospectral Lax representations
and the procedure of converting a nonlinear Lax representation into its linear nonisospectral counterpart
in the (1+2)-dimensional case.
}
\begin{gather}\label{eq:dNLaxPair}
v_t=\frac13\left(v_x^3-\frac{u_{xy}^3}{v_x^3}\right)+u_{xx}v_x-\frac{u_{xy}u_{yy}}{v_x},\quad
v_y=-\frac{u_{xy}}{v_x},
\end{gather}
was derived as a dispersionless counterpart%
\footnote{%
See a technique of limit transitions to dispersionless counterparts of
(1+2)-dimensional differential equations and of the corresponding Lax representations in~\cite[p.~167]{zakh1994a}.
}
of the Lax representation of the Nizhnik equation, cf.~\cite{pavl2006a}.
The maximal Lie invariance (pseudo)algebra~$\mathfrak g_{\rm L}$ of the system~\eqref{eq:dNLaxPair} is spanned by the vector fields
\begin{gather*}
\begin{split}&
\bar D^t(\tau)=\tau\p_t+\tfrac13\tau_tx\p_x+\tfrac13\tau_ty\p_y-\tfrac1{18}\tau_{tt}(x^3+y^3)\p_u,\quad
\bar D^{\rm s}=x\p_x+y\p_y+3u\p_u+\tfrac32v\p_v,\\ &
\bar P^x(\chi)=\chi\p_x-\tfrac12\chi_tx^2\p_u,\quad
\bar P^y(\rho)=\rho\p_y-\tfrac12\rho_ty^2\p_u,\\ &
\bar R^x(\alpha)=\alpha x\p_u,\quad
\bar R^y(\beta)=\beta y\p_u,\quad
\bar Z(\sigma)=\sigma\p_u,\quad
\bar P^v=\p_v,
\end{split}
\end{gather*}
where $\tau$, $\chi$, $\rho$, $\alpha$, $\beta$ and~$\sigma$ are again arbitrary smooth functions of~$t$.
The algebra~$\mathfrak g_{\rm L}$ is infinite-dimensional as the algebra~$\mathfrak g$
and is obtained from~$\mathfrak g$ by extending the vector fields from~$\mathfrak g$
to the additional dependent variable~$v$ and supplementing the extended algebra with the vector field~$\bar P^v$.
The appearance of~$\bar P^v$ is natural and related to the fact that the unknown function~$v$
is defined up to a constant summand.
This is why we can say that the maximal Lie invariance algebra~$\mathfrak g_{\rm L}$ of the system~\eqref{eq:dNLaxPair}
is induced by the maximal Lie invariance algebra~$\mathfrak g$ of the equation~\eqref{eq:dN}.
Up to the antisymmetry of the Lie bracket,
the nonzero commutation relations between vector fields spanning~$\mathfrak g_{\rm L}$
are exhausted by the counterparts of the commutation relations~\eqref{eq:dNCommRelations}
and one more commutation relation involving the vector field~$\bar P^v$,
$[\bar P^v,\bar D^{\rm s}]=\tfrac32\bar P^v$.

Analogously to Lemma~\ref{lem:radicalOfg}, we can prove the following assertion.

\begin{lemma}\label{lem:radicalOfgL}
The radical of~$\mathfrak g_{\rm L}$ is
$\mathfrak r_{\rm L}=\big\langle\bar D^{\rm s},\bar P^x(\chi),\bar P^y(\rho),
\bar R^x(\alpha),\bar R^y(\beta),\bar Z(\sigma),\bar P^v\big\rangle$.
\end{lemma}

Further following the consideration of Section~\ref{sec:LieInvAlgebra},
we construct several megaideals of the algebra~$\mathfrak g_{\rm L}$,
\begin{gather*}
{\mathfrak g_{\rm L}}'
=\big\langle\bar D^t(\tau),\bar P^x(\chi),\bar P^y(\rho),\bar R^x(\alpha),\bar R^y(\beta),\bar Z(\sigma),\bar P^v\big\rangle,\quad
\\
\bar{\mathfrak m}_1:={\mathfrak g_{\rm L}}''
=\big\langle\bar D^t(\tau),\bar P^x(\chi),\bar P^y(\rho),\bar R^x(\alpha),\bar R^y(\beta),\bar Z(\sigma)\big\rangle,\quad
\\
\bar{\mathfrak m}_2:=\mathfrak r_{\rm L}=\big\langle\bar D^{\rm s},\bar P^x(\chi),\bar P^y(\rho),
\bar R^x(\alpha),\bar R^y(\beta),\bar Z(\sigma),\bar P^v\big\rangle,
\\
\bar{\mathfrak m}_2'={\mathfrak g_{\rm L}}'\cap\bar{\mathfrak m}_2
=\big\langle\bar P^x(\chi),\bar P^y(\rho),\bar R^x(\alpha),\bar R^y(\beta),\bar Z(\sigma),\bar P^v\big\rangle,
\\
\bar{\mathfrak m}_3:=\bar{\mathfrak m}_1\cap\bar{\mathfrak m}_2'
=\big\langle\bar P^x(\chi),\bar P^y(\rho),\bar R^x(\alpha),\bar R^y(\beta),\bar Z(\sigma)\big\rangle,
\\
\bar{\mathfrak m}_4:=\bar{\mathfrak m}_2''=\big\langle\bar R^x(\alpha),\bar R^y(\beta),\bar Z(\sigma)\big\rangle,\quad
\bar{\mathfrak m}_5:=\bar{\mathfrak m}_2'''=\big\{\bar Z(\sigma)\big\},
\\
\mathfrak z({\mathfrak g_{\rm L}}')=\big\langle Z(1),\bar P^v\big\rangle,\quad
\bar{\mathfrak m}_6:=\mathfrak z({\mathfrak g_{\rm L}}'')=\big\langle Z(1)\big\rangle,\quad
\bar{\mathfrak m}_7:=\big\langle \bar P^v\big\rangle.
\end{gather*}
The technique for finding the megaideal~$\bar{\mathfrak m}_7$ differs from that for the other obtained megaideals.
This is why we formulate the claim on~$\bar{\mathfrak m}_7$ as an assertion.

\begin{lemma}\label{lem:SpecificMegaidealOfgL}
The span $\bar{\mathfrak m}_7:=\big\langle \bar P^v\big\rangle$ is a megaideal of~$\mathfrak g_{\rm L}$.
\end{lemma}

\begin{proof}
Let $\varphi$ be an arbitrary automorphism of~$\mathfrak g_{\rm L}$.
Since
$\bar D^{\rm s}\in\bar{\mathfrak m}_2\setminus\bar{\mathfrak m}_2'$ and
$\bar P^v\in\mathfrak z({\mathfrak g_{\rm L}}')\setminus\bar{\mathfrak m}_6$,
where
$\bar{\mathfrak m}_2:=\mathfrak r_{\rm L}$,
$\bar{\mathfrak m}_2'$,
$\mathfrak z({\mathfrak g_{\rm L}}')$ and
$\bar{\mathfrak m}_6:=\mathfrak z({\mathfrak g_{\rm L}}'')$ are megaideals of~$\mathfrak g_{\rm L}$,
we have
\begin{gather*}
\varphi(\bar D^{\rm s})=c_0\bar D^{\rm s}
+\bar P^x(\chi^0)+\bar P^y(\rho^0)+\bar R^x(\alpha^0)+\bar R^y(\beta^0)+\bar Z(\sigma^0)+b_0\bar P^v,
\\
\varphi(\bar P^v)=a_1\bar Z(1)+b_1\bar P^v,\quad
\end{gather*}
where $\chi^0$, $\rho^0$, $\alpha^0$, $\beta^0$ and~$\sigma^0$ are smooth functions of~$t$, and
$c_0$, $b_0$, $a_1$ and~$b_1$ are constants with $c_0b_1\ne0$.
Then we evaluate the defining automorphism property at $\varphi$ and
the pair of vector fields $\big(\bar D^{\rm s},\bar P^v\big)$,
\begin{gather*}
\big[\varphi\big(\bar D^{\rm s}),\varphi\big(\bar P^v\big)\big]
=\varphi\big([\bar D^{\rm s},\bar P^v]\big)=-\tfrac32\varphi\big(\bar P^v\big)\quad \sim\quad
(2c_0-1)a_1\bar Z(1)+(c_0-1)b_1\bar P^v=0, 
\end{gather*}
which implies $(2c_0-1)a_1=(c_0-1)b_1=0$.
In view of $b_1\ne0$, we derive $c_0=1$ and then $a_1=0$.
In other words, $\varphi(\bar P^v)\in\langle\bar P^v\rangle=:\bar{\mathfrak m}_7$
for any automorphism~$\varphi$ of~$\mathfrak g_{\rm L}$,
i.e., $\bar{\mathfrak m}_7$ is a megaideal of~$\mathfrak g_{\rm L}$.
\end{proof}

The improper megaideal~$\mathfrak g_{\rm L}$
and the proper megaideals~${\mathfrak g_{\rm L}}'$, $\bar{\mathfrak m}_2'$ and~$\mathfrak z({\mathfrak g_{\rm L}}')$
are inessential in the course of constructing the point-symmetry pseudogroup of the system~\eqref{eq:dNLaxPair}
using the algebraic method since they are sums of other megaideals,
$\mathfrak g_{\rm L}=\bar{\mathfrak m}_1+\bar{\mathfrak m}_2$,
${\mathfrak g_{\rm L}}'=\bar{\mathfrak m}_1+\bar{\mathfrak m}_7$,
$\bar{\mathfrak m}_2'=\bar{\mathfrak m}_3+\bar{\mathfrak m}_7$,
and
$\mathfrak z({\mathfrak g_{\rm L}}')=\bar{\mathfrak m}_6+\bar{\mathfrak m}_7$.
The two constructed one-dimensional megaideals~$\bar{\mathfrak m}_6$ and~$\bar{\mathfrak m}_7$
are the most important for further consideration.
As for the algebra~$\mathfrak g$,
we cannot of course answer the question of whether the above megaideals
exhaust the entire set of proper megaideals of the (infinite-dimensional) algebra~$\mathfrak g_{\rm L}$.

\begin{theorem}\label{thm:dNCompletePointSymGroupOfLaxRepresentation}
The point-symmetry pseudogroup~$G_{\rm L}$ of the nonlinear Lax representation~\eqref{eq:dNLaxPair}
is generated by the transformations of the form
\[
\begin{split}
&\tilde t=T(t),\quad
 \tilde x=A^{2/3}T_t^{1/3}x+X^0(t),\quad
 \tilde y=A^{2/3}T_t^{1/3}y+Y^0(t),\\
&\tilde u=A^2u-\frac{A^2T_{tt}}{18T_t}(x^3+y^3)
 -\frac{A^{4/3}}{2T_t^{1/3}}(X^0_tx^2+Y^0_ty^2)+W^1(t)x+W^2(t)y+W^0(t),\\
&\tilde v=Av+B
\end{split}
\]
and the transformation~$\mathscr J$: $\tilde t=t$, $\tilde x=y$, $\tilde y=x$, $\tilde u=u$, $\tilde v=v$.
Here $T$, $X^0$, $Y^0$, $W^0$, $W^1$ and $W^2$ are arbitrary smooth functions of~$t$
with $T_t\neq0$, and $A$ and~$B$ are arbitrary constants with $A\ne0$.
\looseness=-1
\end{theorem}

\begin{proof}
We follow item~(i) of the proof of Theorem~\ref{thm:dNPointSymPseudogroup},
replacing the point transformation~$\Phi$ by the point transformation~$\bar\Phi$
in the extended space with the coordinates $(t,x,y,u,v)$,
\[
\bar\Phi\colon\ (\tilde t,\tilde x,\tilde y,\tilde u,\tilde v)=(T,X,Y,U,V),
\]
where $(T,X,Y,U,V)$ is a tuple of smooth functions of $(t,x,y,u,v)$ with nonvanishing Jacobian.
The condition $\Phi_*\bar P^v\subseteq\bar{\mathfrak m}_7$ implies that $T_v=X_v=Y_v=U_v=0$ and $V_v=\const$.
Since $\varpi_*\mathfrak g_{\rm L}=\mathfrak g$,
where $\varpi$ in the natural projection from~$\mathbb R^5_{t,x,y,u,v}$ onto~$\mathbb R^4_{t,x,y,u}$,
the independence of $(T,X,Y,U)$ on~$v$ means that
the $t$-, $x$-, $y$- and $u$-components of~$\bar\Phi$ satisfy all the constraints
derived in item~(i) of the proof of Theorem~\ref{thm:dNPointSymPseudogroup}
for the components of the transformation~$\Phi$,
i.e., they have, up to composing with the transformation~$\mathscr J$: $(\tilde t,\tilde x,\tilde y,\tilde u,\tilde v)=(t,y,x,u,v)$,
the form~\eqref{eq:dNPointSymForm}.
It is obvious that the transformation~$\mathscr J$ is a point symmetry of the system~\eqref{eq:dNLaxPair}.
Then collecting $v$-components in the expanded conditions
$\Phi_*\bar Z(1)\in\bar{\mathfrak m}_6$,
$\Phi_*\bar P^z(1)\in\bar{\mathfrak m}_3$, $z\in\{x,y\}$, and
$\Phi_*\bar D^t(1)\in\bar{\mathfrak m}_1$,
lead to the equations $V_t=V_x=V_y=V_u=0$.
Hence it has the form $V=Av+B$, where $A$ and~$B$ are arbitrary constants with $A\ne0$.

Each point transformation~$\Phi$ whose components are of the above form
satisfies the condition $\Phi_*\mathfrak g_{\rm L}\subseteq\mathfrak g_{\rm L}$,
which means that this form cannot be constrained more
within the framework of the purely algebraic method.

We complete the proof with computing by the direct method.
More specifically, using the chain rule,
we derive expressions for derivatives of $(\tilde u,\tilde v)$ with respect to $(\tilde t,\tilde x,\tilde y)$
up to order two in terms of the variables and derivatives without tildes,
successively substitute the obtained expressions
and the expressions for the leading derivatives~$v_t$ and~$v_y$ in view of the system~\eqref{eq:dNLaxPair}
into the system~\eqref{eq:dNLaxPair} written in terms of variables with tildes
and split the derived equations with respect to the other (parametric) derivatives of~$u$ and~$v$ up to order two.
The resulting system of equations for parameters of point symmetry transformations of the system~\eqref{eq:dNLaxPair}
reduces to the single equation $C^3=A^2$, i.e., $C=A^{2/3}>0$.
The computation can be simplified by factoring out the transformation~$\mathscr J$ and the transformations
related to varying the pseudogroup parameters $T$, $X^0$, $Y^0$, $W^0$, $W^1$, $W^2$ and~$B$,
which are obviously point symmetry transformations of~\eqref{eq:dNLaxPair}.
In other words, we can set $T=t$, $X^0=Y^0=W^0=W^1=W^2=0$ and $B=0$.
\end{proof}

\begin{corollary}\label{cor:dNLaxPairDiscrSyms}
A complete list of discrete point symmetry transformations of the system~\eqref{eq:dNLaxPair}
that are independent up to composing with each other and with continuous point symmetry transformations of this system
is exhausted by three commuting involutions, which can be chosen to be
the permutation~$\bar{\mathscr J}$ of the variables~$x$ and~$y$, $(\tilde t,\tilde x,\tilde y,\tilde u,\tilde v)=(t,y,x,u,v)$,
and two transformations~$\bar{\mathscr I}^{\rm i}$ and~$\bar{\mathscr I}^v$ alternating the signs of $(t,x,y)$ and of $v$, respectively,
$\bar{\mathscr I}^{\rm i}\colon(\tilde t,\tilde x,\tilde y,\tilde u,\tilde v)=(-t,-x,-y,u,v)$ and
$\bar{\mathscr I}^v\colon(\tilde t,\tilde x,\tilde y,\tilde u,\tilde v)=(t,x,y,u,-v)$.
\end{corollary}

Therefore, analogously to the pseudogroup~$G$,
the quotient group of the point-symmetry pseudogroup~$G_{\rm L}$
of the nonlinear Lax representation~\eqref{eq:dNLaxPair} of the dispersionless Nizhnik equation~\eqref{eq:dN}
with respect to its identity component is isomorphic to the group $\mathbb Z_2\times\mathbb Z_2\times\mathbb Z_2$.

\begin{remark}\label{rem:ActionOfIsOnLaxRepresentation}
The point transformation~$\bar{\mathscr I}^{\rm s}$: $(\tilde t,\tilde x,\tilde y,\tilde u,\tilde v)=(t,-x,-y,-u,v)$,
which is the trivial extension of the discrete point symmetry transformation~$\mathscr I^{\rm s}$ of the equation~\eqref{eq:dN} to~$v$,
maps the nonlinear Lax representation~\eqref{eq:dNLaxPair} of the equation~\eqref{eq:dN}
to an equivalent nonlinear Lax representation of the same equation,
\[
v_t=-\frac13\left(v_x^3+\frac{u_{xy}^3}{v_x^3}\right)+u_{xx}v_x+\frac{u_{xy}u_{yy}}{v_x},\quad
v_y=\frac{u_{xy}}{v_x}.
\]
\end{remark}

\section{Point-symmetry pseudogroup\\ of dispersionless Nizhnik system}\label{sec:PointSymGroupOfdNSystem}

The equation~\eqref{eq:dN} is in fact a potential equation of the dispersionless counterpart
\begin{gather}\label{eq:dNSystem}
p_t=(h^1p)_x+(h^2p)_y,\quad
h^1_y=p_x,\quad
h^2_x=p_y
\end{gather}
of the original symmetric Nizhnik system \cite[Eq.~(4)]{nizh1980a}, cf.\ footnote~\ref{fnt:dNVariousForms}.
(We re-denote the dependent variables and scale the system variables
for canceling the coefficient~3 on the nonlinear summands
and for setting the constant parameters~$k_1$ and~$k_2$ to~1.)
Indeed, using the last two equations of the system~\eqref{eq:dNSystem} as ``short'' conservation laws,
we introduce the potentials~$\varphi^1$ and~$\varphi^2$ defined by the equations
\[
\varphi^1_x=h^1,\quad \varphi^1_y=p
\quad\mbox{and}\quad
\varphi^2_y=h^2,\quad \varphi^2_x=p.
\]
Therefore, we also have the ``short'' first-level potential conservation law $\varphi^1_y=\varphi^2_x$,
for which the associated second-level potential~$u$ is defined by the equations
with $u_x=\varphi^1$, $u_y=\varphi^2$.%
\footnote{%
See \cite[Section~3.5]{popo2008a} and the end of Section~VI in~\cite{kunz2008a} for related terminology.
The idea of the iterative procedure of introducing potentials can be traced back to~\cite{wahl1975a}.
}
The dependent variables of the system~\eqref{eq:dNSystem} are expressed in terms of the potential~$u$ alone,
$p=u_{xy}$, $h^1=u_{xx}$ and $h^2=u_{yy}$.
Substituting these expressions into the first equation of the system~\eqref{eq:dNSystem},
we derive the equation~\eqref{eq:dN} for the potential~$u$.

Since the equation~\eqref{eq:dN} and the system~\eqref{eq:dNSystem} are related in a nonlocal way,
the maximal Lie invariance (pseudo)algebra~$\mathfrak g_{\rm dN}$ and the point-symmetry pseudogroup~$G_{\rm dN}$
the system~\eqref{eq:dNSystem} cannot be directly derived
from their counterparts~$\mathfrak g$ and~$G$ for the equation~\eqref{eq:dN}
and should be computed independently.
At the same time, each Lie-symmetry vector field~$Q$ of~\eqref{eq:dN}
as belonging to the span of the vector fields~\eqref{eq:dNMIA}
induces a Lie-symmetry vector field~$\hat Q$ of~\eqref{eq:dNSystem}.
The induction map $\mathcal M_*\colon\mathfrak g\to\mathfrak g_{\rm dN}$
is the composition of the standard second prolongation
and the projection from the second-order jet space over the basic space $\mathbb R^3_{t,x,y}\times\mathbb R_u$
onto the space with coordinates $(t,x,y,p,h^1,h^2)$ under the identification $(p,h^1,h^2)=(u_{xy},u_{xx},u_{yy})$.
It is obvious that $\mathcal M_*$ is a Lie-algebra homomorphism with
$\ker\mathcal M_*=\big\langle R^x(\alpha),R^y(\beta),Z(\sigma)\big\rangle$.
The problem is to prove that the homomorphism~$\mathcal M_*$ is surjective,
i.e., it is an epimorphism, $\mathop{\rm im}\mathcal M_*=\mathfrak g_{\rm dN}$.
This is really the case since
computations within the framework of the classical Lie infinitesimal approach show that
\[
\mathfrak g_{\rm dN}=\big\langle\hat D^t(\tau),\hat D^{\rm s},\hat P^x(\chi),\hat P^y(\rho)\big\rangle=\mathcal M_*\mathfrak g,
\]
where
$\hat D^t(\tau)=\mathcal M_*D^t(\tau)$,
$\hat D^{\rm s}=\mathcal M_*D^{\rm s}$,
$\hat P^x(\chi)=\mathcal M_*P^x(\chi)$,
$\hat P^y(\rho)=\mathcal M_*P^y(\rho)$, i.e.,
\begin{gather*}
\hat D^t(\tau)=\tau\p_t+\tfrac13\tau_tx\p_x+\tfrac13\tau_ty\p_y
-\tfrac23\tau_tp\p_p-\tfrac13(2\tau_th^1+\tau_{tt}x)\p_{h^1}-\tfrac13(2\tau_th^2+\tau_{tt}y)\p_{h^2},\\
\hat D^{\rm s}=x\p_x+y\p_y+p\p_p+h^1\p_{h^1}+h^2\p_{h^2},\quad
\hat P^x(\chi)=\chi\p_x-\chi_t\p_{h^1},\quad
\hat P^y(\rho)=\rho\p_y-\rho_t\p_{h^2},
\end{gather*}
and the parameter functions~$\tau$, $\chi$ and $\rho$ run through the set of smooth functions of~$t$.

Since $\mathcal M_*$ is a Lie-algebra epimorphism with
$\ker\mathcal M_*=\big\langle R^x(\alpha),R^y(\beta),Z(\sigma)\big\rangle$,
the nonzero commutation relations between the vector fields spanning~$\mathfrak g_{\rm dN}$
are exhausted, up to the antisymmetry of the Lie bracket, by
\begin{gather}\label{eq:dNSystemCommRelations}
\begin{split}&
[\hat D^t(\tau^1),\hat D^t(\tau^2)]=\hat D^t(\tau^1\tau^2_t-\tau^1_t\tau^2),\\[.5ex]&
[\hat D^t(\tau),\hat P^x(\chi)]=\hat P^x\big(\tau\chi_t-\tfrac13\tau_t\chi\big),\quad
[\hat D^t(\tau),\hat P^y(\rho)]=\hat P^y\big(\tau\rho_t-\tfrac13\tau_t\rho\big),\\[.5ex]&
[\hat D^{\rm s},\hat P^x(\chi)]=-\hat P^x(\chi),\quad
[\hat D^{\rm s},\hat P^y(\rho)]=-\hat P^y(\rho).
\end{split}
\end{gather}
Accordingly, we were able to construct the following proper megaideals of the algebra~$\mathfrak g_{\rm dN}$:
\begin{gather*}
\hat{\mathfrak m}_1:={\mathfrak g_{\rm dN}}'=\big\langle\hat D^t(\tau),\hat P^x(\chi),\hat P^y(\rho)\big\rangle,\\
\hat{\mathfrak m}_2:=\mathfrak r_{\rm dN}=\big\langle\hat D^{\rm s},\hat P^x(\chi),\hat P^y(\rho)\big\rangle,\\
\hat{\mathfrak m}_3:=\hat{\mathfrak m}_2'=\hat{\mathfrak m}_1\cap\hat{\mathfrak m}_2=
\big\langle\hat P^x(\chi),\hat P^y(\rho)\big\rangle.
\end{gather*}

In view of Theorem~\ref{thm:dNPointSymPseudogroup},
the analogous map $\mathcal M\colon G\to G_{\rm dN}$ is a pseudogroup homomorphism,
where $\ker\mathcal M$ is the pseudosubgroup of~$G$
constituted by the transformations of the form~\eqref{eq:dNPointSymForm}
with $T=t$, $C=1$ and $X^0=Y^0=0$.
Although the problem is again to prove the surjection property of~$\mathcal M$,
the presence of this homomorphism allows us to easily make
a conjecture on the general form of point symmetry transformations of the system~\eqref{eq:dNSystem},
which we then prove using the modified version of the megaideal-based method that was suggested in~\cite{malt2021a}.

\begin{theorem}\label{thm:dNSystemPointSymPseudogroup}
The point-symmetry pseudogroup~$G_{\rm dN}$ of the dispersionless Nizhnik system~\eqref{eq:dNSystem}
is generated by the transformations of the form
\begin{gather}\label{eq:dNSystemPointSymForm}
\begin{split}&
\tilde t=T(t),\quad
\tilde x=CT_t^{1/3}x+X^0(t),\quad
\tilde y=CT_t^{1/3}y+Y^0(t),\\&
\tilde p=\frac C{T_t^{2/3}}p,\quad
\tilde h^1=\frac C{T_t^{2/3}}h^1-\frac{CT_{tt}}{3T_t^{5/3}}x-\frac{X^0_t}{T_t},\quad
\tilde h^2=\frac C{T_t^{2/3}}h^2-\frac{CT_{tt}}{3T_t^{5/3}}y-\frac{Y^0_t}{T_t}
\end{split}
\end{gather}
and the transformation~$\hat{\mathscr J}$: $\tilde t=t$, $\tilde x=y$, $\tilde y=x$, $\tilde p=p$, $\tilde h^1=h^2$, $\tilde h^2=h^1$.
Here $T$, $X^0$ and $Y^0$ are arbitrary smooth functions of~$t$ with $T_t\neq0$,
and $C$ is an arbitrary nonzero constant.
\end{theorem}

\begin{proof}
Although the procedure of proving is in general analogous to
that in the proof of Theorem~\ref{thm:dNPointSymPseudogroup},
computational details are essentially different.
Consider a point transformation~$\Phi$ in the space with the coordinates $(t,x,y,p,h^1,h^2)$,
\[
\Phi\colon\ (\tilde t,\tilde x,\tilde y,\tilde p,\tilde h^1,\tilde h^2)=(T,X,Y,P,H^1,H^2),
\]
where $(T,X,Y,P,H^1,H^2)$ is a tuple of smooth functions of $(t,x,y,p,h^1,h^2)$ with nonvanishing Jacobian.
If it is a point symmetry of the system~\eqref{eq:dNSystem},
then the pushforward~$\Phi_*$ of vector fields by~$\Phi$ satisfies the conditions
$\Phi_*\hat{\mathfrak m}_3\subseteq\hat{\mathfrak m}_3$,
$\Phi_*(\hat{\mathfrak m}_1\setminus\hat{\mathfrak m}_3)\subseteq\hat{\mathfrak m}_1\setminus\hat{\mathfrak m}_3$ and
$\Phi_*(\hat{\mathfrak m}_2\setminus\hat{\mathfrak m}_3)\subseteq\hat{\mathfrak m}_2\setminus\hat{\mathfrak m}_3$,
and, moreover, $\ker\Phi_*=\{0\}$.
For evaluating~$\Phi_*$,
we choose the following linearly independent vector fields from~$\mathfrak g$:
\begin{gather*}
Q^{1z}:=\hat P^z(1),\quad Q^{2z}:=\hat P^z(t),\quad Q^{3z}:=\hat P^z(t^2),\\
Q^4:=\hat D^t(1),\quad Q^5:=\hat D^t(t),\quad Q^6:=\hat D^{\rm s}
\end{gather*}
with $z\in\{x,y\}$. Since
$Q^{1z},Q^{2z},Q^{3z}\in\hat{\mathfrak m}_3$,
$Q^4,Q^5\in\hat{\mathfrak m}_1\setminus\hat{\mathfrak m}_3$ and
$Q^6\in\hat{\mathfrak m}_2\setminus\hat{\mathfrak m}_3$,
then
\begin{gather}\label{eq:dNSystemMainPushforwards}
\begin{split}&
\Phi_*Q^{iz}=\tilde P^x(\tilde\chi^{iz})+\tilde P^y(\tilde\rho^{iz}),\quad(\tilde\chi^{iz},\tilde\rho^{iz})\ne(0,0) ,\quad i=1,2,3,
\\&
\Phi_*Q^i=\tilde D^t(\tilde\tau^i)+\tilde P^x(\tilde\chi^i)+\tilde P^y(\tilde\rho^i),\quad \tilde\tau^i\ne0,\quad i=4,5,
\\&
\Phi_*Q^i=\lambda^i\tilde D^{\rm s}+\tilde P^x(\tilde\chi^i)+\tilde P^y(\tilde\rho^i),\quad \lambda^i\ne0,\quad i=6.
\end{split}
\end{gather}
Let \eqref{eq:dNSystemMainPushforwards}$_{iz}$, $i=1,2,3$, and  \eqref{eq:dNSystemMainPushforwards}$_i$, $i=4,5,6$,
refer to the $i$th equation in the system~\eqref{eq:dNSystemMainPushforwards} with $z\in\{x,y\}$ for $i=1,2,3$.

The identity $\Phi_*Q^{3z}-2\Phi_*(t)\Phi_*Q^{2z}+\Phi_*(t^2)\Phi_*Q^{1z}=0$
and the corresponding combination of the equations
\eqref{eq:dNSystemMainPushforwards}$_{1z}$, \eqref{eq:dNSystemMainPushforwards}$_{2z}$ and \eqref{eq:dNSystemMainPushforwards}$_{3z}$
imply the system
\begin{gather*}
\tilde\chi^{3z}(T)-2t\tilde\chi^{2z}(T)+t^2\tilde\chi^{1z}(T)=0,\quad
\tilde\chi^{3z}_{\tilde t}(T)-2t\tilde\chi^{2z}_{\tilde t}(T)+t^2\tilde\chi^{1z}_{\tilde t}(T)=0,
\\
\tilde\rho^{3z}(T)-2t\tilde\rho^{2z}(T)+t^2\tilde\rho^{1z}(T)=0,\quad
\tilde\rho^{3z}_{\tilde t}(T)-2t\tilde\rho^{2z}_{\tilde t}(T)+t^2\tilde\rho^{1z}_{\tilde t}(T)=0,
\end{gather*}
whose differential consequences are
$\tilde\chi^{2z}(T)=t\tilde\chi^{1z}(T)$ and
$\tilde\rho^{2z}(T)=t\tilde\rho^{1z}(T)$.
Similarly to the previous proofs, we derive from these equations that $T=T(t)$ with $T_t\ne0$.
We collect the components in the equations~\eqref{eq:dNSystemMainPushforwards}$_{1z}$
and in the combination $\Phi_*(t)\eqref{eq:dNSystemMainPushforwards}_{1z}-\eqref{eq:dNSystemMainPushforwards}_{2z}$,
deriving the constraints
\begin{gather*}
X_z=\tilde\chi^{1z}(T),\quad
Y_z=\tilde\rho^{1z}(T),\quad
P_z=0,\quad
X_{h^z}=Y_{h^z}=P_{h^z}=0,\\
H^1_z=-\tilde\chi^{1z}_{\tilde t}(T)=-\frac{X_{zt}}{T_t},\quad
H^2_z=-\tilde\rho^{1z}_{\tilde t}(T)=-\frac{Y_{zt}}{T_t},\\
H^1_{h^z}=\tilde\chi^{2z}_{\tilde t}(T)-t\tilde\chi^{1z}_{\tilde t}(T)=\frac{\tilde\chi^{1z}(T)}{T_t}=\frac{X_z}{T_t},\quad
H^2_{h^z}=\tilde\rho^{2z}_{\tilde t}(T)-t\tilde\rho^{1z}_{\tilde t}(T)=\frac{\tilde\rho^{1z}(T)}{T_t}=\frac{Y_z}{T_t}
\end{gather*}
with $h^x:=h^1$ and $h^y:=h^2$.
Hence
\begin{gather}\label{eq:dNSystemPreliminaryFormOfPointSyms}
\begin{split}&
X=X^1(t)x+X^2(t)y+X^0(t,p),\quad
Y=Y^1(t)x+Y^2(t)y+Y^0(t,p),\quad
P=P(t,p),\\&
H^1=\frac{X^1}{T_t}h^1+\frac{X^2}{T_t}h^2-\frac{X^1_t}{T_t}x-\frac{X^2_t}{T_t}y+\breve H^1(t,p),\\&
H^2=\frac{Y^1}{T_t}h^1+\frac{Y^2}{T_t}h^2-\frac{Y^1_t}{T_t}x-\frac{Y^2_t}{T_t}y+\breve H^2(t,p),
\end{split}
\end{gather}
where $X^1$, $X^2$, $X^0$, $Y^1$, $Y^2$, $Y^0$, $P$, $\breve H^1$ and $\breve H^2$
are sufficiently smooth functions of their arguments with $P_p(X^1Y^2-X^2Y^1)\ne0$.

The componentwise splitting of the equations~\eqref{eq:dNSystemMainPushforwards}$_4$,
$3\eqref{eq:dNSystemMainPushforwards}_5-3\Phi_*(t)\eqref{eq:dNSystemMainPushforwards}_4$ and~\eqref{eq:dNSystemMainPushforwards}$_6$
leads to the system
\begin{gather*}
T_t=\tilde\tau^4(T),\quad
X_t= \tfrac13\tilde\tau^4_{\tilde t}(T)X+\tilde\chi^4(T),\quad
Y_t= \tfrac13\tilde\tau^4_{\tilde t}(T)Y+\tilde\rho^4(T),\quad
P_t=-\tfrac23\tilde\tau^4_{\tilde t}(T)P,\\
H^1_t=-\tfrac23\tilde\tau^4_{\tilde t}(T)H^1-\tfrac13\tilde\tau^4_{\tilde t\tilde t}(T)X-\tilde\chi^4_{\tilde t}(T),\quad
H^2_t=-\tfrac23\tilde\tau^4_{\tilde t}(T)H^2-\tfrac13\tilde\tau^4_{\tilde t\tilde t}(T)Y-\tilde\rho^4_{\tilde t}(T),
\\[1ex]
\tilde\tau^5(T)=t\tilde\tau^4(T),\quad
xX_x+yX_y-2pX_p=X+3\tilde\chi^5(T)-3t\tilde\chi^4(T),\\
xY_x+yY_y-2pY_p=Y+3\tilde\rho^5(T)-3t\tilde\rho^4(T),\quad
pP_p=P,\\
xH^1_x+yH^1_y-2pH^1_p-2h^1H^1_{h^1}-2h^2H^1_{h^2}=-2H^1+(T_t^{-1})_tX-3\tilde\chi^5_{\tilde t}(T)+3t\tilde\chi^4_{\tilde t}(T),\\
xH^2_x+yH^2_y-2pH^2_p-2h^1H^2_{h^1}-2h^2H^2_{h^2}=-2H^2+(T_t^{-1})_tY-3\tilde\rho^5_{\tilde t}(T)+3t\tilde\rho^4_{\tilde t}(T),
\\[1ex]
\noprint{
tT_t=\tilde\tau^5(T),\quad
tX_t+\tfrac13xX_x+\tfrac13yX_y-\tfrac23pX_p= \tfrac13\tilde\tau^5_{\tilde t}(T)X+\tilde\chi^5(T),\\
tY_t+\tfrac13xY_x+\tfrac13yY_y-\tfrac23pY_p= \tfrac13\tilde\tau^5_{\tilde t}(T)Y+\tilde\rho^5(T),\quad
tP_t-\tfrac23pP_p=-\tfrac23\tilde\tau^5_{\tilde t}(T)P,\\
tH^1_t+\tfrac13xH^1_x+\tfrac13yH^1_y-\tfrac23pH^1_p-\tfrac23h^1H^1_{h^1}-\tfrac23h^2H^1_{h^2}
=-\tfrac23\tilde\tau^5_{\tilde t}(T)H^1-\tfrac13\tilde\tau^5_{\tilde t\tilde t}(T)X-\tilde\chi^5_{\tilde t}(T),\\
tH^2_t+\tfrac13xH^2_x+\tfrac13yH^2_y-\tfrac23pH^2_p-\tfrac23h^1H^2_{h^1}-\tfrac23h^2H^2_{h^2}
=-\tfrac23\tilde\tau^5_{\tilde t}(T)H^2-\tfrac13\tilde\tau^5_{\tilde t\tilde t}(T)Y-\tilde\rho^5_{\tilde t}(T).
\\[1ex]
}
xX_x+yX_y+pX_p=\lambda^6X+\tilde\chi^6(T),\quad
xY_x+yY_y+pY_p=\lambda^6Y+\tilde\rho^6(T),\quad
pP_p=\lambda^6P,\\
xH^1_x+yH^1_y+pH^1_p+h^1H^1_{h^1}+h^2H^1_{h^2}=\lambda^6H^1-\tilde\chi^6_{\tilde t}(T),
\\
xH^2_x+yH^2_y+pH^2_p+h^1H^2_{h^1}+h^2H^2_{h^2}=\lambda^6H^2-\tilde\rho^6_{\tilde t}(T).
\end{gather*}
Here we at once take into account the constraints $\tilde\tau^4(T)=T_t$ and $\tilde\tau^5(T)=tT_t$,
in view of which we have
$\tilde\tau^4_{\tilde t}(T)=T_{tt}/T_t$,
$\tilde\tau^5_{\tilde t}(T)-t\tilde\tau^4_{\tilde t}(T)=1$,
$\tilde\tau^4_{\tilde t\tilde t}(T)=(T_{tt}/T_t)_t/T_t$ and
$\tilde\tau^5_{\tilde t\tilde t}(T)-t\tilde\tau^4_{\tilde t\tilde t}(T)=-(T_t^{-1})_t$. 
We substitute the earlier derived form~\eqref{eq:dNSystemPreliminaryFormOfPointSyms} of the components of~$\Phi$ into the above system,
split the expanded system with respect to~$(x,y,h^1,h^2)$ and solve the obtained system of constraints
\begin{gather*}
X^j_t=\frac{T_{tt}}{3T_t}X^j,\ \
Y^j_t=\frac{T_{tt}}{3T_t}Y^j,\ \ j=1,2,\quad
X^0_t=\frac{T_{tt}}{3T_t}X^0+\tilde\chi^4(T),\quad
Y^0_t=\frac{T_{tt}}{3T_t}Y^0+\tilde\rho^4(T),
\\
P_t=-\frac{2T_{tt}}{3T_t}P,\ \,
\breve H^1_t=-\frac{2T_{tt}}{3T_t}\breve H^1\!-\!\left(\!\frac{T_{tt}}{T_t}\!\right)_{\!\!t}\!\frac{X^0}{3T_t}-\tilde\chi^4_{\tilde t}(T),\ \,
\breve H^2_t=-\frac{2T_{tt}}{3T_t}\breve H^2\!-\!\left(\!\frac{T_{tt}}{T_t}\!\right)_{\!\!t}\!\frac{Y^0}{3T_t}-\tilde\rho^4_{\tilde t}(T),
\\
pX^0_p=X^0+\tilde\chi^6(T),\quad
-2pX^0_p=X^0+3\tilde\chi^5(T)-3t\tilde\chi^4(T),\quad
\\
pY^0_p=Y^0+\tilde\rho^6(T),\quad
-2pY^0_p=Y^0+3\tilde\rho^5(T)-3t\tilde\rho^4(T),\quad
\quad
pP_p=P,\quad \lambda^6=1,\quad
\\
p\breve H^1_p=\breve H^1-\tilde\chi^6_{\tilde t}(T),\quad
3xH^1_x+3yH^1_y=(T_t^{-1})_tX-2\tilde\chi^6_{\tilde t}(T)-3\tilde\chi^5_{\tilde t}(T)+3t\tilde\chi^4_{\tilde t}(T),
\\
p\breve H^2_p=\breve H^2-\tilde\rho^6_{\tilde t}(T),\quad
3xH^2_x+3yH^2_y=(T_t^{-1})_tY-2\tilde\rho^6_{\tilde t}(T)-3\tilde\rho^5_{\tilde t}(T)+3t\tilde\rho^4_{\tilde t}(T).
\end{gather*}
This system implies that in fact $X^0=X^0(t)$ and $Y^0=Y^0(t)$.
The other its independent consequences are only
\begin{gather*}
\tilde\chi^4(T)=X^0_t-\frac{T_{tt}}{3T_t}X^0,\quad
\tilde\rho^4(T)=Y^0_t-\frac{T_{tt}}{3T_t}Y^0,
\\
3\big(\tilde\chi^5(T)-t\tilde\chi^4(T)\big)=\tilde\chi^6(T)=-X^0,\quad
3\big(\tilde\rho^5(T)-t\tilde\rho^4(T)\big)=\tilde\rho^6(T)=-Y^0,
\\
X^j_t=\frac{T_{tt}}{3T_t}X^j,\ \
Y^j_t=\frac{T_{tt}}{3T_t}Y^j,\ \ j=1,2,\quad
P_t=-\frac{2T_{tt}}{3T_t}P,\quad
pP_p=P,
\\
p\breve H^1_p=\breve H^1+\frac{X^0_t}{T_t},\quad
\breve H^1_t=-\frac{2T_{tt}}{3T_t}\breve H^1-\frac{X^0_{tt}}{T_t}+\frac{T_{tt}}{3T_t^2}X^0_t,
\\
p\breve H^2_p=\breve H^2+\frac{Y^0_t}{T_t},\quad
\breve H^2_t=-\frac{2T_{tt}}{3T_t}\breve H^2-\frac{Y^0_{tt}}{T_t}+\frac{T_{tt}}{3T_t^2}Y^0_t.
\noprint{
\\
(T_t^{-1})_tX^0-2\tilde\chi^6_{\tilde t}(T)-3\tilde\chi^5_{\tilde t}(T)+3t\tilde\chi^4_{\tilde t}(T)\equiv0,
\\
(T_t^{-1})_tY^0-2\tilde\rho^6_{\tilde t}(T)-3\tilde\rho^5_{\tilde t}(T)+3t\tilde\rho^4_{\tilde t}(T)\equiv0.
}
\end{gather*}
The equations for the parameter functions involved in $\Phi$ integrate to
\begin{gather}\label{eq:dNSystemPreliminaryFormOfPointSyms2}
X^j=A_jT_t^{1/3},\ \
Y^j=B_jT_t^{1/3},\ \
P=\frac{Cp}{T_t^{2/3}},\ \
\breve H^1=\frac{E_1p}{T_t^{2/3}}-\frac{X^0_t}{T_t},\ \
\breve H^2=\frac{E_2p}{T_t^{2/3}}-\frac{Y^0_t}{T_t},
\end{gather}
where $A_j$, $B_j$, $C$ and $E_j$, $j=1,2$, are constants with $C(A_1B_2-A_2B_1)\ne0$.
\noprint{
\begin{gather*}
\tilde t=T,\quad
\tilde x=T_t^{1/3}(A_1x+A_2y)+X^0,\quad
\tilde y=T_t^{1/3}(B_1x+B_2y)+Y^0,\quad
\tilde p=\frac{Cp}{T_t^{2/3}},\\
\tilde h^1=\frac{A_1h^1+A_2h^1+E_1p}{T_t^{2/3}}-\frac{T_{tt}}{3T_t^{5/3}}(A_1x+A_2y)-\frac{X^0_t}{T_t},
\\
\tilde h^2=\frac{B_1h^1+B_2h^1+E_2p}{T_t^{2/3}}-\frac{T_{tt}}{3T_t^{5/3}}(B_1x+B_2y)-\frac{Y^0_t}{T_t}.
\end{gather*}
}

The transformations defined by~\eqref{eq:dNSystemPreliminaryFormOfPointSyms}--\eqref{eq:dNSystemPreliminaryFormOfPointSyms2}
constitute a pseudogroup~$\mathfrak G$,
which contains the set~$\mathfrak N$ of the transformations of the form~\eqref{eq:dNSystemPointSymForm} with $C=1$
as a normal pseudosubgroup.
The later transformations are point symmetries of the system~\eqref{eq:dNSystem},
which can be easily checked by the direct method
although it is also clear due to the observation that they are generated by Lie symmetries of the system~\eqref{eq:dNSystem}
and the time reflection $(\tilde t,\tilde x,\tilde y,\tilde p,\tilde h^1,\tilde h^2)=(-t,-x,-y,p,h^1,h^2)$.
The pseudogroup~$\mathfrak G$ splits over~$\mathfrak N$, $\mathfrak G=\mathfrak H\ltimes\mathfrak N$,
where the subgroup~$\mathfrak H$ of~$\mathfrak G$ consists of the transformations of the form
\begin{gather*}
T=t,\quad
X=A_1x+A_2y,\quad
Y=B_1x+B_2y,\\
P=Cp,\quad
H^1=A_1h^1+A_2h^2+E_1p,\quad
H^2=B_1h^1+B_2h^2+E_2p,\quad
\end{gather*}
$A_j$, $B_j$, $C$ and $E_j$, $j=1,2$, are arbitrary constants with $C(A_1B_2-A_2B_1)\ne0$.
Therefore, we can factor out the transformations from~$\mathfrak N$ and
consider only the transformations from~$\mathfrak H$ in the remainder of the proof.

It is easy to check that $\Psi_*\mathfrak g_{\rm dN}=\mathfrak g_{\rm dN}$ for any $\Psi\in\mathfrak H$.
This means that no constraints for the above constant parameters can be found within the algebraic approach.
Therefore, for completing the proof, the direct method should necessarily be applied.
The computation is standard.
The chain rule implies expressions for first-order derivatives
of $(\tilde p,\tilde h^1,\tilde h^2)$ with respect to $(\tilde t,\tilde x,\tilde y)$
in terms of the variables and derivatives without tildes,
which we substitute jointly with the expressions for $(\tilde p,\tilde h^1,\tilde h^2)$
and, e.g., the expressions for the derivatives~$p_t$, $h^1_y$ and~$h^2_x$ according to the system~\eqref{eq:dNSystem}
into the system~\eqref{eq:dNSystem} written in terms of variables with tildes
and split the derived equations with respect to the other (parametric) first-order derivatives of~$(p,h^1,h^2)$.
As a result, we derive the system
\begin{gather*}
A_1A_2=B_1B_2=0,\quad
B_1E_1=A_1E_2,\quad
B_2E_1=A_2E_2,\quad
\\
A_1^{\,2}-E_1A_2=CB_2,\quad
A_2^{\,2}-E_1A_1=CB_1,\quad
B_1^{\,2}-E_2B_2=CA_2,\quad
B_2^{\,2}-E_2B_1=CA_1.
\end{gather*}
In view of the inequality $C(A_1B_2-A_2B_1)\ne0$,
it implies that $E_1=E_2=0$ and $(A_1,A_2,B_1,B_2)\in\big\{(C,0,0,C),\,(0,C,C,0)\big\}$.
\end{proof}

\begin{corollary}\label{cor:dNInductionOfAlgAndGroup}
$\mathcal M_*\mathfrak g=\mathfrak g_{\rm dN}$ and $\mathcal MG=G_{\rm dN}$.
In other words,
the maximal Lie invariance algebra~$\mathfrak g_{\rm dN}$ and the point-symmetry pseudogroup~$G_{\rm dN}$
the system~\eqref{eq:dNSystem} are induced by
their counterparts~$\mathfrak g$ and~$G$ for the equation~\eqref{eq:dN}, respectively.
\end{corollary}

\begin{corollary}\label{cor:dNSystemDiscrSyms}
A complete list of discrete point symmetry transformations of the system~\eqref{eq:dNSystem}
that are independent up to composing with each other and with continuous point symmetry transformations of this equation
is exhausted by three commuting involutions, which can be chosen to be
the permutation~$\mathscr J$ of the variables~$x$ and~$y$,
and two transformations~$\mathscr I^{\rm i}$ and~$\mathscr I^{\rm s}$
alternating the signs of $(t,x,y)$ and of $(x,y,p,h^1,h^2)$, respectively,
\begin{gather*}
\mathscr J\colon(\tilde t,\tilde x,\tilde y,\tilde p,\tilde h^1,\tilde h^2)=(t,y,x,p,h^2,h^1),\\[1ex]
\mathscr I^{\rm i}\colon(\tilde t,\tilde x,\tilde y,\tilde p,\tilde h^1,\tilde h^2)=(-t,-x,-y,p,h^1,h^2),\\[1ex]
\mathscr I^{\rm s}\colon(\tilde t,\tilde x,\tilde y,\tilde p,\tilde h^1,\tilde h^2)=(t,-x,-y,-p,-h^1,-h^2).
\end{gather*}
\end{corollary}

Hence again the quotient group of the point-symmetry pseudogroup~$G_{\rm dN}$ of the dispersionless Nizhnik system~\eqref{eq:dNSystem}
with respect to the identity component of this pseudogroup is isomorphic to the group $\mathbb Z_2\times\mathbb Z_2\times\mathbb Z_2$.

\section{Defining geometric properties}\label{sec:DefiningGeometricProperties}

We find geometric properties of the dispersionless Nizhnik equation~\eqref{eq:dN}
that completely define this equation.
In this section, by $u_\kappa$ or by $u_{\kappa_0\kappa_1\kappa_2}$
with the multi-index $\kappa=(\kappa_0,\kappa_1,\kappa_2)\in\mathbb N_0^{\,\,3}$
we denote the jet variable that is associated with the derivative
$\p^{\kappa_0+\kappa_1+\kappa_2}u/\p t^{\kappa_0}\p x^{\kappa_1}\p y^{\kappa_2}$.

\begin{lemma}\label{lem:dNEqsWithSameIA}
A partial differential equation of order less than or equal to three with three independent variables
is invariant with respect to the algebra~$\mathfrak g$ if and only if it is of the form
\begin{gather}\label{eq:dNEqsWithSameIA}
u_{txy}=(u_{xx}u_{xy})_x+(u_{xy}u_{yy})_y+u_{xy}u_{xyy}H\left(\frac{u_{xxx}-u_{yyy}}{u_{xyy}},\frac{u_{xxy}}{u_{xyy}}\right),
\end{gather}
where $H$ is an arbitrary smooth function of its arguments.
\end{lemma}

\begin{proof}
The (infinite) prolongations~$Q_{(\infty)}$ of the vector fields~$Q$,
which are presented in~\eqref{eq:dNMIA}
and span the maximal Lie invariance (pseudo)algebra~$\mathfrak g$ of the equation~\eqref{eq:dN},
are
\begin{gather*}
D^t_{(\infty)}(\tau)=\tau\p_t+\tfrac13\tau_tx\p_x+\tfrac13\tau_ty\p_y
\\ \hphantom{D^t_{(\infty)}(\tau)=}
-\sum_{k=2}^\infty\tau^{(k)}\left(
 \frac{x^3+y^3}{18}\p_{u_{k-2,00}}
+\frac{x^2}6\p_{u_{k-2,10}}+\frac x3\p_{u_{k-2,20}}+\frac13\p_{u_{k-2,30}}\right.
\\ \hphantom{D^t_{(\infty)}(\tau)=-\sum_{k=2}^\infty\tau^{(k)}}\ \left.
{}+\frac{y^2}6\p_{u_{k-2,01}}+\frac y3\p_{u_{k-2,02}}+\frac13\p_{u_{k-2,03}}\right)
\\ \hphantom{D^t_{(\infty)}(\tau)=}
-\sum_\kappa\sum_{k=1}^\infty\tau^{(k)}\left(
\binom{\kappa_0}k+\frac{\kappa_1+\kappa_2}3\binom{\kappa_0}{k-1}
\right)u_{\kappa_0+1-k,\kappa_1\kappa_2}\p_{u_\kappa}
\\ \hphantom{D^t_{(\infty)}(\tau)=}
-\frac13\sum_\kappa\sum_{k=2}^\infty\binom{\kappa_0}{k-1}\tau^{(k)}
\big(xu_{\kappa_0+1-k,\kappa_1+1,\kappa_2}+yu_{\kappa_0+1-k,\kappa_1,\kappa_2+1}\big)\p_{u_\kappa},
\\[2ex]
D^{\rm s}_{(\infty)}=x\p_x+y\p_y+\sum_\kappa(3-\kappa_1-\kappa_2)u_\kappa\p_{u_\kappa},
\\[1ex]
P^x_{(\infty)}(\chi)=\chi\p_x-\sum_{k=1}^\infty\chi^{(k)}\!\left(
\frac{x^2}2\p_{u_{k-1,00}}+x\p_{u_{k-1,10}}+\p_{u_{k-1,20}}
+\sum_\kappa\binom{\kappa_0}ku_{\kappa_0-k,\kappa_1+1,\kappa_2}\p_{u_\kappa}
\right)\!,
\\[1ex]
P^y_{(\infty)}(\rho)=\rho\p_y-\sum_{k=1}^\infty\rho^{(k)}\!\left(
\frac{y^2}2\p_{u_{k-1,00}}+y\p_{u_{k-1,01}}+\p_{u_{k-1,02}}
+\sum_\kappa\binom{\kappa_0}ku_{\kappa_0-k,\kappa_1,\kappa_2+1}\p_{u_\kappa}
\right)\!,
\\[1ex]
R^x_{(\infty)}(\alpha)=\sum_{k=0}^\infty\alpha^{(k)}\big(x\p_{u_{k00}}+\p_{u_{k10}}\big),\quad
R^y_{(\infty)}(\beta)=\sum_{k=0}^\infty\beta^{(k)}\big(y\p_{u_{k00}}+\p_{u_{k01}}\big),
\\[1ex]
Z_{(\infty)}(\sigma)=\sum_{k=0}^\infty\sigma^{(k)}\p_{u_{k00}}.
\end{gather*}
Recall that  run through the set of smooth functions of~$t$,
$\binom nk:=0$ if $k>n$.
The differential invariants of the identity component~$G_{\rm id}$
of the point-symmetry pseudogroup~$G$ of the equation~\eqref{eq:dN}
can be found as differential functions~$F$ of~$u$ that satisfy the equations $Q_{(\infty)}F=0$,
where $Q$ runs through the set of the vector fields~\eqref{eq:dNMIA}.
We can split these equations
with respect to the derivatives of the parameter functions~$\tau$, $\chi$, $\rho$, $\alpha$, $\beta$ and~$\sigma$
and then, after deriving the equations $F_x=F_y=0$, with respect to~$x$ and~$y$,
which leads to the following system for~$F$:
\begin{gather*}
F_t=F_x=F_y=F_{u_{k00}}=F_{u_{k10}}=F_{u_{k01}}=0,\\[1ex]
F_{u_{k20}}+\sum_\kappa\binom{\kappa_0}{k+1}u_{\kappa_0-k-1,\kappa_1+1,\kappa_2}F_{u_\kappa}=0,\ \
F_{u_{k02}}+\sum_\kappa\binom{\kappa_0}{k+1}u_{\kappa_0-k-1,\kappa_1,\kappa_2+1}F_{u_\kappa}=0,\\[1ex]
\sum_\kappa(3-\kappa_1-\kappa_2)u_\kappa F_{u_\kappa}=0,\ \
\sum_\kappa(\kappa_0+1)u_\kappa F_{u_\kappa}=0, 
\\
F_{u_{k30}}+F_{u_{k03}}+\sum_\kappa\left(
3\binom{\kappa_0}{k+2}+(\kappa_1+\kappa_2)\binom{\kappa_0}{k+1}
\right)u_{\kappa_0-k-1,\kappa_1\kappa_2}F_{u_\kappa}=0, \quad k\in\mathbb N_0.
\end{gather*}

Let the order of~$F$ as a differential function be less than or equal to three.
Then the equations in the first row mean that $F$ is a function at most
$u_{xx}$, $u_{xy}$, $u_{yy}$, $u_{txx}$, $u_{txy}$, $u_{tyy}$, $u_{xxx}$, $u_{xxy}$, $u_{xyy}$ and $u_{yyy}$.
Then the equations in the second row with $k=1$ and with $k=0$ successively imply
\[F_{u_{txx}}=F_{u_{tyy}}=0,\quad F_{u_{xx}}+u_{xxy}F_{u_{txy}}=0,\quad F_{u_{yy}}+u_{xyy}F_{u_{txy}}=0.\]
The equations in the third row and the equation in the last row with $k=0$ reduce to
\begin{gather*}
u_{xx}F_{u_{xx}}+u_{xy}F_{u_{xy}}+u_{yy}F_{u_{yy}}+u_{txy}F_{u_{txy}}=0,\\
u_{xxx}F_{u_{xxx}}+u_{xxy}F_{u_{xxy}}+u_{xyy}F_{u_{xyy}}+u_{yyy}F_{u_{yyy}}+u_{txy}F_{u_{txy}}=0,\\
F_{u_{xxx}}+F_{u_{yyy}}+2u_{xy}F_{u_{txy}}=0.
\end{gather*}
The other equations are satisfied identically in view of the derived equations.
Integrating the latter equations, we obtain that any differential invariant~$F$ of order less than or equal to three
of the group~$G_{\rm id}$ is a function of
\[
\omega_0:=\frac{u_{txy}-(u_{xx}u_{xy})_x-(u_{xy}u_{yy})_y}{u_{xy}u_{xyy}},\quad
\omega_1:=\frac{u_{xxx}-u_{yyy}}{u_{xyy}},\quad
\omega_2:=\frac{u_{xxy}}{u_{xyy}},
\]
$F=F(\omega_0,\omega_1,\omega_2)$.
Hence the group~$G_{\rm id}$ admits no differential invariants of orders zero, one and two.
The above consideration also implies that the group~$G_{\rm id}$
admits no codimension-one singular invariant manifolds
in the third-order jet space ${\rm J}^3(\mathbb R^3_{txy}\times\mathbb R_u)$.
Therefore, a partial differential equation for the unknown function~$u$ depending on $(t,x,y)$ is $G_{\rm id}$-invariant
if and only if it is of the form $F(\omega_0,\omega_1,\omega_2)=0$,
where $F_{\omega_0}\ne0$ since otherwise the variable~$t$ is not significant and rather plays the role of a parameter,
and \eqref{eq:dNEqsWithSameIA} is an equivalent form for such equations.
\end{proof}

Any equation of the form~\eqref{eq:dNEqsWithSameIA} is invariant with respect to the point transformations
$\mathscr I^{\rm i}$ and~$\mathscr I^{\rm s}$ alternating the signs of $(t,x,y)$ and of $(x,y,u)$, respectively,
cf.\ Corollary~\ref{cor:dNDiscrSyms}.
At the same time, the permutation~$\mathscr J$ of the variables~$x$ and~$y$ is a point symmetry transformation
of such an equation if and only if \[H(\omega_1,\omega_2)=\omega_2H(-\omega_2^{-1}\omega_1,\omega_2^{-1}).\]

The space of local conservation laws of the dispersionless Nizhnik equation~\eqref{eq:dN}
is infinite-dimensional, and the simplest conservation-law characteristics of this equation
are $1$, $u_{xx}$ and~$u_{yy}$.
Let us check when an equation of the form~\eqref{eq:dNEqsWithSameIA} admits these conservation-law characteristics.

\begin{lemma}\label{lem:dNEqsWithSameIAandCLChars}
(i) An equation of the form~\eqref{eq:dNEqsWithSameIA} admits the conservation-law characteristic~$1$
and thus it is in conserved form if and only if
$H$ is an affine function of~$(\omega_1,\omega_2)$,
i.e., $H=a\omega_1+b\omega_2+c$ for some constants~$a$, $b$ and~$c$,
and the equation takes the form
\begin{gather}\label{eq:dNEqsWithSameIAandCLChar1}
u_{txy}=(u_{xx}u_{xy})_x+(u_{xy}u_{yy})_y+u_{xy}\big(a(u_{xxx}-u_{yyy})+bu_{xxy}+cu_{xyy}\big).
\end{gather}

(ii) An equation of the form~\eqref{eq:dNEqsWithSameIAandCLChar1}
admits the conservation-law characteristic~$u_{xx}$ or~$u_{yy}$ if and only if
$a=b=0$ or~$a=c=0$, respectively.
\end{lemma}

\begin{proof}
The differential function \[N:=u_{txy}-(u_{xx}u_{xy})_x-(u_{xy}u_{yy})_y-u_{xy}u_{xyy}H(\omega_1,\omega_2)\]
is (locally) a total divergence if and only if $\mathsf EN=0$,
where $\mathsf E$ is the Euler operator with respect to~$u$,
$\mathsf E:=\sum_\kappa(-\mathrm D_t)^{\kappa_0}(-\mathrm D_x)^{\kappa_1}(-\mathrm D_y)^{\kappa_2} \p_{u_\kappa}$,
see, e.g., \cite[Theorem 4.7]{olve1993A}.
Collecting coefficients of sixth-order derivatives of~$u$ in the equation $\mathsf EN=0$,
we derive the system \[H_{\omega_1\omega_1}=H_{\omega_1\omega_2}=H_{\omega_2\omega_2}=0,\]
and, in view of this system, the equation $\mathsf EN=0$ is satisfied identically.
This proves item (i).

To prove item (ii), it suffices to similarly consider the equations $\mathsf E(u_{xx}N)=0$ and $\mathsf E(u_{yy}N)=0$
for affine functions~$H$ of~$(\omega_1,\omega_2)$.
\end{proof}

Lemmas~\ref{lem:dNEqsWithSameIA} and~\ref{lem:dNEqsWithSameIAandCLChars} jointly imply the following theorem.

\begin{theorem}\label{thm:dNDefiningGeometricProperties}
An $r$th order ($r\in\{1,2,3\}$) partial differential equation with three independent variables
admits the algebra~$\mathfrak g$ as its Lie invariance algebra
and the conservation-law characteristics~$1$, $u_{xx}$ and~$u_{yy}$
if and only if it coincides with the dispersionless Nizhnik equation~\eqref{eq:dN}.
\end{theorem}

In view of Theorem~\ref{thm:dNDefiningGeometricProperties},
the invariance with respect to the algebra~$\mathfrak g$ and
admitting the conservation-law characteristics~$1$, $u_{xx}$ and~$u_{yy}$
lead to the invariance with respect to the entire group~$G$,
which includes the discrete point symmetry transformations~$\mathscr J$, $\mathscr I^{\rm i}$ and~$\mathscr I^{\rm s}$,
and to admitting the entire (infinite-dimensional) space of conservation-law characteristics
of the equation~\eqref{eq:dN}.

\section{Discussion}\label{sec:Discussion}

Let us discuss some implications of the paper's results in the form of a chain of remarks.

\begin{remark}\label{rem:dNAlgMethodFirstApplicationToContSymGroups}
Item (ii) of the proof of Theorem~\ref{thm:dNPointSymPseudogroup} is the first example
of applying the megaideal-based version of the algebraic method to computing
the contact-symmetry (pseudo)group of a partial differential equation in the literature.
An example of computing contact symmetries of a partial differential equation
using the automorphism-based version of the algebraic method was presented in~\cite{hydo1998b}.
\end{remark}

\begin{remark}\label{rem:dNAlgMethodSpecificFeature}
Item (i) of the proof of Theorem~\ref{thm:dNPointSymPseudogroup} shows
that the conditions~\eqref{eq:dNMainPushforwards} exhaustively define
the point-symmetry pseudogroup~$G$ of the equation~\eqref{eq:dN},
which is the first example of such a kind in the literature.
In other words, the second part of the computation procedure of the algebraic method
using the direct method is a trivial check that all the singled out point transformations,
which are either of the form~\eqref{eq:dNPointSymForm}
or compositions of transformations of the form~\eqref{eq:dNPointSymForm} with the transformation~$\mathscr J$,
are indeed symmetries of the equation~\eqref{eq:dN}.
In view of item (ii) of the proof of Theorem~\ref{thm:dNPointSymPseudogroup},
the same claim is relevant for the contact-symmetry pseudogroup~$G_{\rm c}$ of the equation~\eqref{eq:dN} as well.
At the same time, this is not the case for the point-symmetry pseudogroup~$G_{\rm L}$
of the nonlinear Lax representation~\eqref{eq:dNLaxPair} of the equation~\eqref{eq:dN}
and even more so for the point-symmetry pseudogroup~$G_{\rm dN}$ of the system~\eqref{eq:dNSystem},
which is nonlocally related to~\eqref{eq:dN}.
\end{remark}

\begin{remark}\label{rem:dNNecessityOfComputingCompletePointSymGroup}
It is obvious that $\mathscr J$, $\mathscr I^{\rm i}$ and~$\mathscr I^{\rm s}$ are point symmetries of the equation~\eqref{eq:dN},
and the identity component of the pseudogroup~$G$, whose infinitesimal counterpart is the algebra~$\mathfrak g$,
consists of the point transformations of the form~\eqref{eq:dNPointSymForm} with $T_t>0$ and $C>0$.
Therefore, all the transformations described in Theorem~\ref{thm:dNPointSymPseudogroup}
are point symmetries of~\eqref{eq:dN},
and the first prolongation of these transformations are contact symmetries of~\eqref{eq:dN}.
At the same time, this is a simple part of the statement of Theorem~\ref{thm:dNPointSymPseudogroup}
although it is still not too trivial as shown by the imprecise formulation of its analog in~\cite{moro2021a}.
In fact, the purpose of the proof of Theorem~\ref{thm:dNPointSymPseudogroup}
is to check that the equation~\eqref{eq:dN} admits no other point and contact symmetry transformations.
\end{remark}

\begin{remark}\label{rem:dNMinSetOfMegaIdeals}
As noted at the end of Section~\ref{sec:LieInvAlgebra},
the nonzero improper megaideal of~$\mathfrak g$, which is the entire algebra~$\mathfrak g$ itself,
can be neglected in the course of applying the megaideal-based method
to computing the point-symmetry pseudogroup~$G$ of the equation~\eqref{eq:dN}
since it is the sum of two proper megaideals, $\mathfrak g=\mathfrak m_1+\mathfrak m_2$.
This is not the case for the megaideals~$\mathfrak m_1$ and~$\mathfrak m_2$.
Nevertheless, if we use one of them, then the condition $\Phi_*\mathfrak m\subseteq\mathfrak m$
for the other implies no new constraints for the transformation components,
and the megaideal set $\{\mathfrak m_2,\dots,\mathfrak m_6\}$ assures a bit more effective and simpler computations
than $\{\mathfrak m_1,\mathfrak m_3,\dots,\mathfrak m_6\}$.
It is not yet clear how to a priori identify megaideals whose involvement in the computation
is not too essential although they are not sums of other proper megaideals.
\end{remark}

\begin{remark}\label{rem:dNSpanningSubalg}
The span of each of the sets of linearly independent vector fields that were selected
for use in the course of applying the megaideal-based version of the algebraic method in question
in the present paper and in~\cite{malt2021a}
is closed with respect to the Lie bracket, i.e.,
it is a subalgebra of the corresponding invariance algebra.
It is still not known whether this property plays a certain role
and whether its appearance is an occasional phenomenon or
appropriate vector fields can be always chosen in the way to possess it.
\end{remark}

\begin{remark}\label{rem:dNAlgMinimalRequiredSetOfVFs}
The selected sets of linearly independent vector fields are unexpectedly small
but still allow us to effectively compute the corresponding point- and contact-symmetry groups,
especially when involving megaideals.
Nevertheless, we do not know whether the cardinalities of these sets are minimum.
In general, one has developed no techniques that would help
to a priori estimate sufficient numbers of such vector fields,
not to mention finding the minimum among these numbers
and the optimal selection of vector fields for simplifying computations.
\end{remark}

\begin{remark}\label{rem:dNAlgMethodSufficentNumberOfVFs}
In the course of computing the point-symmetry pseudogroup~$G$ of the equation~\eqref{eq:dN},
it is optimal and sufficient to use the conditions
$\Phi_*(\mathfrak s_1\cap\mathfrak m_j)\subseteq\mathfrak m_j$, $j=2,\dots,6$,
which jointly implies the condition $\Phi_*\mathfrak g\subseteq\mathfrak g$.
We have additionally checked that for each $k\in\{2,\dots,6\}$,
the collection of the conditions
$\Phi_*(\mathfrak s_1\cap\mathfrak m_j)\subseteq\mathfrak m_j$, $j\in M_k$,
where $M_6=\{6\}$, $M_5=\{5\}$, $M_4=\{4,5\}$, $M_3=\{3,4,5\}$, $M_2=\{2,3,4,5\}$ and $M_1=\{2,3,4,5,6\}$,
implies the condition $\Phi_*\mathfrak m_k\subseteq\mathfrak m_k$.
As noted in Remark~\ref{rem:dNMinSetOfMegaIdeals},
using the subalgebra~$\mathfrak s_2$ leads to a bit more complicated computations
but allows us to replace the megaideal~$\mathfrak m_2$ with~$\mathfrak m_1$.
Moreover, it suffices to consider the conditions
$\Phi_*(\mathfrak s_2\cap\mathfrak m_j)\subseteq\mathfrak m_j$, $j=1,3,4,5$,
which jointly implies the conditions
$\Phi_*\mathfrak g\subseteq\mathfrak g$, and thus
$\Phi_*\mathfrak m_k\subseteq\mathfrak m_k$, $k\in\{1,2,6\}$, as well.
\end{remark}

\begin{remark}\label{rem:dNAlgMethodPropertyP}
In the course of proving Theorem~\ref{thm:dNDefSubalgs},
we have checked which subalgebras of~$\mathfrak g$ among~$\mathfrak s_1$, $\bar{\mathfrak s}_1$ and~$\mathfrak s_2$
define diffeomorphisms stabilizing~$\mathfrak g$.
Simultaneously, we have in fact recomputed the group~$G$ only using the condition $\Phi_*\mathfrak s_2\subseteq\mathfrak g$,
i.e., involving no proper megaideals of~$\mathfrak g$.
Although the recomputation is based on a technique
analogous to that in the proof of Theorem~\ref{thm:dNPointSymPseudogroup}, it is much more complicated.
Moreover, in contrast to the condition with~$\mathfrak s_2$,
the analogous condition with~$\bar{\mathfrak s}_1$ or, moreover, with~$\mathfrak s_1$
does not imply the complete system of determining equations
for point symmetries of the equation~\eqref{eq:dN}.
The situation is even more dramatic in the case of contact transformations.
From the condition $\Psi_*(\mathfrak s_4\cap\mathfrak m_{j(1)})\subseteq\mathfrak m_{j(1)}$, $j=4,5$,
for a contact transformation~$\Psi$,
where $\mathfrak s_4=\langle Z(1),Z(t),R^x(1),R^y(1)\rangle$ is the common four-dimensional subalgebra
of~$\mathfrak s_1$ and~$\mathfrak s_2$, it is easy to derive
that this transformation is the first prolongation
of a point transformation in the space with the coordinates $(t,x,y,u)$.
This is the content of item (ii) of the proof of Theorem~\ref{thm:dNPointSymPseudogroup}.
We do not know whether this property of~$\Psi$ follows even from the condition
$\Psi_*(\mathfrak s_1\cup\mathfrak s_2)_{(1)}\subseteq\mathfrak g_{(1)}$.
This weakened condition, which does not involve the knowledge of megaideals of~$\mathfrak g$,
implies a too complicated system of equations for the components of~$\Psi$,
and techniques applied in the present paper are not appropriate for solving such a system.
This is the reason why in Theorem~\ref{thm:dNContDefSubalgs} we have used the subalgebra~$\mathfrak s_3$,
which is wider than the subalgebra~$\mathfrak s_4$.
The presented facts demonstrate the importance of using proper megaideals
within the framework of the algebraic method for finding point-symmetry groups
of systems of differential equations.
\end{remark}

\begin{remark}\label{rem:dNComparingSymsOfdNAndOfItslaxRepresentation}
In contrast to continuous point symmetry transformations,
not all discrete point symmetry transformations of the equation~\eqref{eq:dN}
are extended to ones of its nonlinear Lax representation~\eqref{eq:dNLaxPair}.
At the same time, the system~\eqref{eq:dNLaxPair} admits,
in addition to the expectable point symmetries of simple shifts in~$v$,
the discrete point symmetry transformation alternating the sign of~$v$.
\end{remark}

\begin{remark}\label{rem:dNDefiningGeomProperties}
Although the maximal Lie invariance algebra~$\mathfrak g$ of the equation~\eqref{eq:dN} exhaustively defines
the point-symmetry pseudogroup~$G$ of this equation,
it does not define exhaustively the equation itself.
Nevertheless, to single out the equation~\eqref{eq:dN} from the entire set of
third-order partial differential equations with three independent variables,
it suffices to supplement the $\mathfrak g$-invariance with a few nice conditions.
As such conditions, we have selected
admitting the conservation-law characteristics~$1$, $u_{xx}$ and~$u_{yy}$.
\end{remark}

The enhanced description of the point- and contact-symmetry pseudogroups of the dispersionless Nizhnik equation~\eqref{eq:dN}
is only the first step in further enhancing the results of~\cite{moro2021a} on Lie reductions of this equation.
We also plan to reclassify the one-dimensional subalgebras of the algebra~$\mathfrak g$ and
to classify its two-dimensional subalgebras,
to exhaustively carry out Lie reductions of the equation~\eqref{eq:dN} and then to accurately study its hidden Lie symmetries.
In addition, we would like to compute the entire algebra of (local) generalized symmetries of the equation~\eqref{eq:dN}
and the entire space of its local conservation laws.
We conjecture that in fact,
the generalized symmetries of the equation~\eqref{eq:dN} are exhausted by its Lie symmetries,
and the essential order of its conservation-law characteristics is not greater than two.
The above results will create necessary prerequisites for the consideration of
nonlocal symmetry-like objects that are related to the equation~\eqref{eq:dN}.\looseness=-1

A~similar study can be carried out for both the symmetric and asymmetric Nizhnik equations
over the real and the complex fields in the presence of dispersion~\cite{nizh1980a,vese1984a},
for the Nizhnik equation in the Novikov--Veselov form and its dispersionless counterpart,
as well as for the stationary Nizhnik equation,
which was considered in \cite{fera1999c,marv2003a} and in \cite[Sections~9.7~and~9.8]{roge2002A}.

\section*{Acknowledgments}

The authors are grateful to Dmytro Popovych, Galyna Popovych and Artur Sergyeyev for helpful discussions and suggestions.
This work was supported in part by a grant from the Simons Foundation (1030291, O.O.V., V.M.B.)
and by the Ministry of Education, Youth and Sports of the Czech Republic (M\v SMT \v CR)
under RVO funding for I\v C47813059 (R.O.P.).
R.O.P. expresses his gratitude for the hospitality shown by the University of Vienna during his long staying at the university.
The authors express their deepest thanks to the Armed Forces of Ukraine and the civil Ukrainian people
for their bravery and courage in defense of peace and freedom in Europe and in the entire world from russism.


\end{document}